\documentclass[runningheads,orivec]{llncs}
\usepackage{cite}
\usepackage{subcaption} 
\usepackage{amsmath,amssymb,amsfonts}
\usepackage{thm-restate}
\usepackage{booktabs}
\usepackage[ruled,linesnumbered]{algorithm2e}
\SetKwRepeat{Do}{do}{while}
\SetKwComment{tcp}{$\rhd\;$}{}%
\SetKw{True}{true}
\usepackage{placeins}
\usepackage{graphicx}
\usepackage[export]{adjustbox}
\usepackage{textcomp}
\usepackage{xcolor}
\usepackage{tcolorbox}
\usepackage{xspace}
\usepackage{multirow}
\usepackage{tikz}
\usetikzlibrary{automata, positioning, arrows.meta, calc}

\usepackage{caption}
\usepackage[breaklinks]{hyperref}
\hypersetup{
colorlinks = true, 
linkcolor=blue, 
citecolor=blue, 
urlcolor=blue 
}

\usepackage{orcidlink}
\usepackage[misc,geometry]{ifsym}

\usepackage{enumerate}
\usepackage{enumitem}
\makeatletter
\def\namedlabel#1#2{\begingroup
    #2%
    \def\@currentlabel{#2}%
    \phantomsection\label{#1}\endgroup
}
\makeatother

\iftrue
    \usepackage{todonotes}
    
    \newcommand{\aj}[1]{\todo[inline,color=blue!10,caption={AJ}]{\textbf{AJ:} #1}}
    \newcommand{\at}[1]{\todo[inline,color=teal!10,caption={AT}]{\textbf{AT:} #1}}
    \newcommand{\ly}[1]{\todo[inline,color=orange!10,caption={LY}]{\textbf{LY:} #1}}
    \newcommand{\mjj}[1]{\todo[inline,color=purple!10,caption={MJj}]{\textbf{Mjj:} #1}}
\else
    \newcommand{\aj}[1]{}
    \newcommand{\at}[1]{}
    \newcommand{\ly}[1]{}
    \newcommand{\mjj}[1]{}
\fi

\newcommand{\wordletter}[2]{#1{[#2]}}

\newcommand{\emptyword}{\varepsilon}

\newcommand{\finwords}{\alphabet^*}

\newcommand{\setnocond}[1]{\{#1\}}

\newcommand{\A}{\mathcal{A}}

\newcommand{\T}{\mathcal{T}}

\newcommand{\pref}{\mathsf{prefixes}}

\newcommand{\trans}{\delta}
\newcommand{\init}{\iota}
\newcommand{\alphabet}{\Sigma}
\newcommand{\states}{Q}

\newcommand{\benchexec}{\textsc{BenchExec}\xspace}
\newcommand{\toolOriginal}{\textsc{ParetoAPTA}\xspace}
\newcommand{\toolLaufferOurEncoding}{\textsc{Pareto3DFA}\xspace}
\newcommand{\toolOurStateOptimalDIP}{\textsc{StatesOptimalDIP}\xspace}

\newcommand{\setcond}[2]{\{\, #1 \mid #2 \,\}}

\newcommand{\nats}{\mathbb{N}}

\newcommand{\lang}{\mathcal{L}}

\newcommand{\proc}[1]{\textsc{#1}}

\newcommand{\entropy}{\mathcal{E}}

\newcommand{\myparagraph}[1]{\medskip\noindent{\bf #1.}}

\newcommand{\oomit}[1]{}

\renewcommand{\epsilon}{\varepsilon}

\allowdisplaybreaks

\begin{document}

\title{Efficient Decomposition Identification of Deterministic Finite Automata from Examples}

\titlerunning{Efficient Decomposition Identification of DFAs from Examples}
%

\author{
Junjie Meng \inst{1}\orcidlink{0009-0009-0603-1899} \and
Jie An \inst{2}$^{\text{(\Letter)}}$\orcidlink{0000-0001-9260-9697} \and
Yong Li \inst{3}$^{\text{(\Letter)}}$\orcidlink{0000-0002-7301-9234} \and
Andrea Turrini \inst{3}\orcidlink{0000-0003-4343-9323} \and
Fanjiang Xu \inst{2}\orcidlink{0009-0004-6016-7360} \and \\
Naijun Zhan \inst{4}\orcidlink{0000-0003-3298-3817} \and
Miaomiao Zhang \inst{1}$^{\text{(\Letter)}}$\orcidlink{0000-0001-9179-0893}
}
\authorrunning{J. Meng et al.}
%
\institute{
School of Computer Science and Technology, Tongji University, Shanghai, China \and
National Key Lab. of Space Integrated Information System, Institute of Software, Chinese Academy of Sciences, Beijing, China  \and
Key Lab. of System Software (Chinese Academy of Sciences) and State Key
Lab. of Computer Science, Institute of Software, Chinese Academy of Sciences, Beijing, China \and
School of Computer Science \& MOE Key Laboratory of High Confidence Software Technologies, Peking University, Beijing, China 
\\ \email{anjie@iscas.ac.cn, liyong@ios.ac.cn, miaomiao@tongji.edu.cn}
}

\maketitle              
\begin{abstract}
The identification of deterministic finite automata (DFAs) from labeled examples is a cornerstone of automata learning, yet traditional methods focus on learning monolithic DFAs, which often yield a large DFA lacking simplicity and interoperability. Recent work addresses these limitations by exploring DFA decomposition identification problems (DFA-DIPs), which model system behavior as intersections of multiple DFAs, offering modularity for complex tasks. However, existing DFA-DIP approaches depend on SAT encodings derived from Augmented Prefix Tree Acceptors (APTAs), incurring scalability limitations due to their inherent redundancy.

In this work, we advance DFA-DIP research through studying two variants: the traditional Pareto-optimal DIP and the novel states-optimal DIP, which prioritizes a minimal number of states. We propose a novel framework that bridges DFA decomposition with recent advancements in automata representation. One of our key innovations replaces APTA with 3-valued DFA (3DFA) derived directly from labeled examples. This compact representation eliminates redundancies of APTA, thus drastically reducing variables in the improved SAT encoding. Experimental results demonstrate that our 3DFA-based approach achieves significant efficiency gains for the Pareto-optimal DIP while enabling a scalable solution for the states-optimal DIP.

\keywords{DFA decomposition \and DFA identification \and Model learning \and Passive learning \and SAT solving \and Grammatical inference.}
\end{abstract}

\setcounter{footnote}{0}

\section{Introduction} \label{sec:introduction}

The identification of Deterministic Finite Automata (DFAs) from labeled examples is a fundamental problem in computer science, with applications in inference of network invariants~\cite{DBLP:conf/cade/GrinchteinLP06}, grammatical inference~\cite{DBLP:journals/pr/Higuera05}, model checking~\cite{DBLP:conf/atva/Neider12}, reinforcement learning~\cite{LaufferYVSS22}, etc.
Known as passive model learning~\cite{Vaandrager17}, classical methods focus on inferring a single DFA from examples. The generated DFA may have a large size and thus suffers from a lack of simplicity and interpretability, as a single DFA representing complex system behaviors can have a very intricate structure. 

To address this issue, recent research~\cite{LaufferYVSS22} has moved toward the \emph{DFA decomposition identification problem} (DFA-DIP), i.e., inferring a group of DFAs from examples, where their conjunction language includes all positive examples and excludes all negative examples. 
This approach allows for capturing sub-tasks performed by the system, with system behavior being described as the intersection of the languages from several DFAs, thereby improving interpretability.

Existing DFA identifying approaches~\cite{HeuleV10,ZakirzyanovSU17,ZakirzyanovMIUM19} typically first employ a data structure called the Augmented Prefix Tree Acceptor (APTA)~\cite{CosteN98} to represent the labeled examples and then encode the identification of the minimal DFA from the APTA as a satisfiability (SAT) problem for Boolean formulas.
The number of Boolean variables required in the SAT problem grows polynomially with the size of the constructed APTA.
However, the size of the APTA increases dramatically with both the number and the length of the examples, since every prefix of the examples corresponds to a unique state in the APTA. 
This results in a corresponding increase in the number of Boolean variables and, consequently, the size of the SAT problem, posing a significant challenge to SAT solvers. 

To alleviate this challenge, a recent work~\cite{DellErbaLS24} extends the technique presented in~\cite{DaciukMWW00} to construct the minimal 3-valued DFA (3DFA) consistent with the given examples, i.e., the 3DFA accepts all positive examples and rejects all negative examples. According to its definition, APTA can be viewed as a specific kind of 3DFA. The constructed minimal 3DFA is dramatically smaller than the original APTA. Hence, the proposed method via minimal 3DFAs in~\cite{DellErbaLS24} significantly improves the DFA identification.

Since the state-of-the-art algorithm for identifying DFA decompositions~\cite{LaufferYVSS22} still relies on the basic APTA construction, a natural improvement would be to apply this minimal 3DFA construction to further reduce the number of Boolean variables required.
However, our findings suggest that the minimal 3DFA construction from~\cite{DellErbaLS24} \emph{cannot} be directly applied to the current DFA decomposition learning framework~\cite{LaufferYVSS22}.
This is because with a minimal 3DFA, its structural characteristics are different from those of the prefix tree, which makes the original encoding no longer applicable (cf.~Section~\ref{sec:encoding}).

\myparagraph{Contributions} To advance the DFA decomposition identification research, we make several contributions in this paper, as summarized below:
\begin{itemize}  
    \item We review the Pareto-optimal DIP studied in~\cite{LaufferYVSS22} and introduce a novel DFA-DIP, named the states-optimal DIP, that prioritizes decompositions with smaller state spaces. (Section~\ref{sec:problems})
    \item We propose a new method for constructing a succinct 3DFA consistent with the given examples and an improved SAT encoding via 3DFA, both tailored for DFA-DIPs. (Section~\ref{sec:encoding})
    \item We propose a method for solving the novel states-optimal DIP. (Section~\ref{sec:overall-algo})
    \item We have implemented the improved method for solving the Pareto-optimal DIP by replacing the original encoding via APTA in~\cite{LaufferYVSS22} with our improved encoding via 3DFA, and have also implemented our method for solving the states-optimal DIP. The experimental results show that our new 3DFA construction significantly reduces the number of states compared to using APTA and dramatically improves the efficiency in solving the Pareto-optimal DIP. Additionally, we present preliminary experiments on the novel states-optimal DIP, highlighting the scalability of our method. (Section~\ref{sec:experiments})
\end{itemize}


\myparagraph{Related work} We review the most related work on DFA identification from examples. 
The most common approach is the evidence driven state-merging (EDSM) algorithm~\cite{DBLP:conf/icgi/LangPP98}. It first constructs an APTA consistent with the given examples, and then iteratively applies a state-merging procedure until
no valid merges are left. However, the issue of this algorithm is that it terminates at a local optimum. 
Current SAT-based methods typically construct an APTA first, and then reduce the problem to a SAT solving problem~\cite{HeuleV10}, with its help. These approaches return a minimal DFA.
Several works~\cite{ZakirzyanovMIUM19,UlyantsevZS15,UlyantsevZS16,ZakirzyanovSU17} have improved the second step by proposing symmetry-breaking techniques and compact SAT encoding.
More recently, the first step, which is an APTA construction, has been improved in~\cite{DellErbaLS24} by extending the technique in~\cite{DaciukMWW00} to construct a minimal 3DFA consistent with the given examples, instead of an APTA. The minimal 3DFA requires fewer states than APTA. 
However, all these works only consider learning a \emph{single} DFA from the examples.
In contrast to them, \cite{LaufferYVSS22} introduces DFA-DIP and extends the SAT encoding via APTA to infer DFA decompositions under their so-called Pareto-optimal partial order.
As the minimal 3DFA cannot directly be used for DFA-DIP, our work proposes a method to construct a non-minimal 3DFA with fewer states than APTA and an improved SAT encoding for DFA-DIP. 
All the works above, including ours, focus on learning \emph{unknown} DFA decompositions from examples.
When the finite-state automata to decompose are known, specific approaches~\cite{Ashar1992,Egri-Nagy11,KupfermanM15} can be used. 

\myparagraph{Outline} 
After reviewing the basic definitions in Section~\ref{sec:preliminaries}, we give the formal definitions of the two DFA-DIPs in Section~\ref{sec:problems}. 
We then introduce our improved SAT encoding via 3DFA in Section~\ref{sec:encoding} and our method for solving the states-optimal DIP and the Pareto-optimal DIP in Section~\ref{sec:overall-algo}. 
We present our experimental evaluation in Section~\ref{sec:experiments} and some concluding remarks in Section~\ref{sec:conclusion}.

In what follows, due to the limited space available, the omitted proofs of lemmas and theorems are given in Appendix~\ref{appendix:proofs}.

%
\section{Preliminaries} \label{sec:preliminaries}


In this paper, given $n \in \nats$, we denote by $[n]$ the set $\setnocond{1, 2, \dots, n}$. 
We fix a finite alphabet $\alphabet$ of letters. 
A \emph{word} $u$ is a finite sequence of letters in $\Sigma$. 
We denote by $\epsilon$ the empty word and with $\Sigma^{*}$ the set of all finite words, and let $\Sigma^{+} = \Sigma^{*} \setminus \{\epsilon\}$. 
Given a word $u$, we denote by $|u|$ the \emph{length} of $u$ ($|\epsilon| = 0$) and by $\wordletter{u}{i}$ the $i$-th letter of $u$ for $0 \leq i < |u|$. 
Given two words $u$ and $v$, we denote by $u \cdot v$ (for short $uv$) their concatenation. 
We say that a word $u$ is a \emph{prefix} of a word $w$ if $w = u \cdot v$ for some word $v \in \Sigma^{*}$. 
We denote by $\pref(u)$ the set of all prefixes of $u$ and we extend it to a set of words $S$ in the usual way, i.e., $\pref(S) = \bigcup_{u \in S} \pref(u)$.

\myparagraph{Transition system}
A \emph{deterministic} transition system (TS) is a tuple $\T = (\states, \init, \trans)$, where $\states$ is a finite set of states, $\init \in \states$ is the initial state, and $\trans \colon \states \times \alphabet \to \states$ is a transition function. 
We also extend $\trans$ from letters to words in the usual way, by letting $\trans(q, \emptyword) = q$ and $\trans(q, a \cdot u) = \trans(\trans(q, a), u)$, where $u \in \finwords$ and $a \in \alphabet$.
The \emph{run} of a TS $\T$ on a finite word $u$ of length $n$ is the sequence of states $\rho = q_{0} q_{1} \cdots q_{n} \in \states^{+}$ such that, for every $0 \leq i < n$, $q_{i+1} = \trans(q_{i}, \wordletter{u}{i})$.

\begin{definition}[Deterministic finite automata] \label{def:dfa}
    A deterministic finite automaton (DFA) is a tuple $\A = (\T, A)$, where $\T$ is a deterministic TS and $A\subseteq Q$ is a set of \emph{accepting} states. 
\end{definition}

It is easy to extend DFAs to process languages with ``don't-care'' words. 

\begin{definition}[3-valued DFAs~\cite{DellErbaLS24}] \label{def:3dfa}
    A 3-valued DFA (3DFA) is a triple $\A = (\T, A, R)$, where $\T$ is a deterministic TS and $A$, $R$, and $D = \states \setminus (A \cup R)$ partition the set of states $\states$, where $A \subseteq \states$ is the set of \emph{accepting} states, $R \subseteq \states$ is the set of \emph{rejecting} states, and the remaining states $D$ are called \emph{don't-care} states.
\end{definition}
A run is \emph{accepting} (respectively, \emph{rejecting}) if it ends in an accepting (resp.\@ rejecting) state.
A finite word $u \in \finwords$ is \emph{accepted} (resp.\@ \emph{rejected}) by $\A$ if it has an accepting (resp.\@ rejecting) run on $u$.
3DFAs map all words in $\finwords$ to \emph{three} values:
accepting ($+$), rejecting ($-$), and don't-care ($?$), where they are accepting if they have an accepting run, rejecting if they have a rejecting run, and don't-care otherwise.
Note that DFAs are a special type of 3DFAs with only accepting and rejecting states.
We denote the language of $\A$ by $\lang(\A)$, i.e., the set of words accepted by $\A$.

\section{Problem Formalization} \label{sec:problems}

In this work we consider two specific DFA decomposition identification problems (DIPs). 
One is the Pareto-optimal DIP established in~\cite{LaufferYVSS22}, and the other is the states-optimal DIP we introduce in this paper. 
Before formally presenting the two DIPs, we recall the definition of DFA decompositions.

\begin{definition}[DFA decomposition~\cite{LaufferYVSS22}] \label{def:dfa_decomposition}
    Let $\mathcal{H}$ be the set of all DFAs over $\Sigma$. 
    A $(m_{1}, \dots, m_{n})$-DFA decomposition is a tuple of $n$ DFAs $(\A_{1}, \dots, \A_{n}) \in \mathcal{H}^n$ such that each DFA $\A_{i}$ has $m_{i}$ states and $m_{1} \leq m_2 \leq \dots \leq m_{n}$.
\end{definition}

A decomposition $(\A_{1}, \dots, \A_{n})$ accepts a word $u$ if and only if all DFAs accept $u$, i.e., $u \in \lang(\A_{i})$ for all $1 \leq i \leq n$; a word that is not accepted is rejected.
Therefore its language $\lang(\A_{1}, \dots, \A_{n})$ is the intersection of the languages of the individual DFAs, i.e., $\lang(\A_{1}, \dots, \A_{n}) = \bigcap_{1 \leq i \leq n} \lang(\A_{i})$. We say $(m_{1},\dots,m_{n})$ is the \emph{states allocation} of the DFA decomposition $(\A_{1}, \dots, \A_{n})$.

\myparagraph{DFA decomposition identification problem (DFA-DIP)} 
Given a set of labeled examples $S = (S^{+}, S^{-})$ where $S^{+}$ contains positive examples and $S^{-}$ contains negative examples, respectively, a general DFA-DIP asks to identify a DFA decomposition consistent with the example set $S$; 
formally, it asks to find a DFA decomposition $(\A_{1}, \dots, \A_{n})$ for a given integer $n \in \nats$ that satisfies:
\begin{description}
\item[\namedlabel{itm:consistency}{C1 Consistency}:] $S^{+} \subseteq \mathcal{L}(\A_{1}, \dots, \A_{n})$ and $S^{-} \subseteq \Sigma^{*} \setminus \mathcal{L}(\A_{1}, \dots, \A_{n})$.
\end{description}
Obviously, the identification problem of monolithic DFAs is a special case of DFA-DIP with $n=1$, i.e., learning a single DFA from examples.
We say a DFA decomposition $(\A_{1}, \cdots, \A_{n})$ is \emph{consistent} with $S$ if it satisfies \hyperref[itm:consistency]{\textbf{C1}}.

To compare decompositions, a \emph{Pareto-optimal partial order} $\prec$ has been introduced in~\cite{LaufferYVSS22}. 
Formally, given a $(m_{1}, \dots, m_{n})$-DFA decomposition $(\A_{1}, \dots, \A_{n})$ and a $(m'_{1}, \dots, m'_{n})$-DFA decomposition $(\A'_{1}, \dots, \A'_{n})$, we say $(\A_{1}, \dots, \A_{n}) \prec (\A'_{1},\dots,\A'_{n})$ if $m_{i} \leq m'_{i}$ for all $1 \leq i \leq n$ and there exists $1 \leq j \leq n$ such that $m_{j} < m'_{j}$.
If so, we say that $(\A_{1}, \dots, \A_{n})$ \emph{dominates} $(\A'_{1}, \dots, \A'_{n})$. 

As incomparable decompositions exist, they can form a Pareto frontier of solutions to the Pareto-optimal DIP defined as follows:

\begin{tcolorbox}[boxrule=1pt,colback=white,colframe=black!75, left = 1pt, right = 1pt, top = 0pt, bottom = 0pt]
\small
\noindent\textbf{Pareto-Optimal DIP~\cite{LaufferYVSS22}.} Given a set of labelled examples $S = (S^{+}, S^{-})$, and an integer $n \in \nats$, find a $(m_{1}, \dots, m_{n})$-DFA decomposition $(\A_{1}, \dots, \A_{n})$ such that (i) it satisfies \hyperref[itm:consistency]{\textbf{C1}}, and (ii) there does not exist a decomposition satisfying \hyperref[itm:consistency]{\textbf{C1}} dominating it under the Pareto-optimal partial order $\prec$.
\end{tcolorbox}


In this paper, we also introduce a novel DFA-DIP --- states-optimal DIP --- which depends on the number of states and the \emph{entropy} of a decomposition defined as follows. 

\begin{definition}[Entropy of DFA decomposition] \label{def:entropy}
Given a $(m_{1}, \dots, m_{n})$-DFA decomposition $(\A_{1}, \dots, \A_{n})$, its entropy is defined as
\[
    \entropy(\A_{1}, \dots, \A_{n}) = -\sum_{i=1}^{n} P(i)\log_2 P(i)
\]
where $P(i) = m_{i}/\sum_{j=1}^{n} m_{j}$ for $1 \leq i \leq n$.
\end{definition}




Based on Definition~\ref{def:entropy}, we present a \emph{states-optimal preorder} as follows. 
Given a $(m_{1}, \dots, m_{n})$-DFA decomposition $(\A_{1}, \dots, \A_{n})$ and a $(m'_{1}, \dots, m'_l)$-DFA decomposition $(\A'_{1}, \dots, \A'_l)$, we say $(\A_{1}, \dots, \A_{n}) \lessdot (\A'_{1}, \dots, \A'_l)$ if $\sum_{i=1}^{n} m_{i} < \sum_{j=1}^{l} m'_{j}$, or $\entropy(\A_{1}, \dots, \A_{n}) \geq \entropy(\A'_{1}, \dots, \A'_l)$ if $\sum_{i=1}^{n} m_{i} = \sum_{j=1}^{l} m'_{j}$. 
The preorder reflects that we prefer a DFA decomposition that has the minimal number of states or, for the same number of states, it contains more individual DFAs with similar (smaller) size. 

A decomposition with a higher entropy value indicates a more evenly distributed state allocation among DFAs, leading to a more flexible and potentially more generalizable representation of the input data. 
In contrast, lower entropy values suggest that the state allocation is concentrated in fewer DFAs, which may reduce the system’s ability to capture diverse structural variations in the data. 
Note that the states-optimal preorder does not require two DFA decompositions having the same number of DFAs, as does the Pareto-optimal partial order. Therefore, the states-optimal preorder is a total preorder. 


\begin{tcolorbox}[boxrule=1pt,colback=white,colframe=black!75, left = 1pt, right = 1pt, top = 0pt, bottom = 0pt]
\small
\noindent\textbf{States-Optimal DIP.} Given a set of labelled examples $S = (S^{+}, S^{-})$, find $(m_{1}, \dots, m_{n})$-DFA decomposition $(\A_{1}, \dots, \A_{n})$ for some $n \in \nats$ such that it is a minimal decomposition w.r.t.\@ the states-optimal preorder $\lessdot$ satisfying \hyperref[itm:consistency]{\textbf{C1}}.
\end{tcolorbox}

Note that the states-optimal DIP does not require a given input $n$ to restrict the number of DFAs in the resulted DFA decomposition. 
Of course, a specific variant is to find a states-optimal decomposition with a given input $n$.

\section{An Improved SAT Encoding of DFA-DIP via 3DFAs} \label{sec:encoding}
In this section, we first briefly review the existing encoding via APTA proposed in~\cite{LaufferYVSS22}, then present our improved SAT encoding of DFA-DIP via 3DFA. 
We describe a challenge in utilizing minimal 3DFAs~\cite{DellErbaLS24} within our encoding at last.

\subsection{Existing encoding via APTAs} 
\label{subsec:apta_encoding}

In~\cite{LaufferYVSS22}, Lauffer \textit{et~al.} proposed a SAT encoding method for DFA-DIP by extending the SAT encoding for identifying a single DFA from examples~\cite{HeuleV10,UlyantsevZS15,UlyantsevZS16,ZakirzyanovMIUM19}.
Their encoding reduces DFA-DIP into a graph coloring problem where the graph is in fact an APTA, a tree-based 3DFA.
In the graph coloring problem, the state of each node of the tree structure, each edge of the tree, and for each state, each transition of the DFA, can be represented by a different color variable~\cite{HeuleV10}.

In an APTA, each state corresponds to a unique prefix of a word in the set of examples $S = (S^{+}, S^{-})$ that we also write as $S = S^{+} \cup S^{-}$. 
First, one can define a function $f \colon \pref(S) \to \nats_S $ where $\nats_S = \{0, \dots, |\pref(S)|-1\}$. 
Intuitively, $f$ maps each prefix $u \in \pref(S)$ to a state in the APTA represented by a unique number in $\nats_S$.
Formally, an APTA $\mathcal{P}$ of $S$ is a 3DFA $(\T, A, R )$ where the TS $\T$ consists of the set of states $\nats_S$, the initial state $f(\epsilon)$, and the transition function $\delta$ defined as 
$\delta(i, a) = j$ if $f(u) = i$ and $f(ua) = j$ where $u, ua \in \pref(S)$ and $a \in \Sigma$;
we define $A = \setcond{i \in \nats_S}{\text{$f(u) = i$ for some $u \in S^{+}$}}$ and $R = \setcond{i \in \nats_S}{\text{$f(u) = i$ for some $u \in S^{-}$}}$.
For example, Fig.~\hyperref[subfig:example_apta]{1(a)} shows the APTA generated from the set of examples $S = (S^{+}, S^{-})$ where $S^{+} = \{aab, aaa, ab\}$, $S^{-} = \{b, aba\}$, and $f = \{\epsilon \mapsto 0, a \mapsto 1, aa \mapsto 2, aab \mapsto 3, aaa \mapsto 4, ab \mapsto 5, aba \mapsto 6, b \mapsto 7\}$.
The set of accepting states is $A = \{3,4,5\}$ and the set of rejecting states is $R = \{6,7\}$. 

\begin{figure}[t] %
  \centering
  \begin{minipage}[b]{0.6\textwidth} 
    \centering
    \begin{subfigure}[b]{\textwidth}
    \centering
    \begin{tikzpicture}[
        node distance=2cm, 
        every state/.style={draw, thick, circle, minimum size=0.7cm}, 
        accepting/.style={double}, 
        false node/.style={draw, thick, rectangle, minimum size=0.6cm} 
      ]
      \node[state] (0) {0};
      \node[state] (1) [below right=0.2cm and 1cm of 0] {1};
      \node[false node] (7) [above right=0.2cm and 1cm of 0] {7};
      \node[state] (2) [below right=0.2cm and 1cm of 1] {2};
      \node[state,accepting] (5) [above right=0.2cm and 1cm of 1] {5};
      \node[state, accepting] (3) [below right=0.2cm and 1cm of 2] {3};
      \node[state, accepting] (4) [above right=0.2cm and 1cm of 2] {4};
      \node[false node] (6) [above right=0.2cm and 1cm of 5] {6};
      \path[->,>=Stealth,thick]
      (0) edge node [pos=0.5, above] {$a$} (1)
      (0) edge node [pos=0.5, above] {$b$} (7)
      (1) edge node [pos=0.5, above]  {$a$} (2)
      (1) edge node [pos=0.5, above] {$b$} (5)
      (2) edge node [pos=0.5, above]  {$b$} (3)
      (2) edge node [pos=0.5, above] {$a$} (4)
      (5) edge node [pos=0.5, above] {$a$} (6);
    \end{tikzpicture}
    \subcaption{The APTA consistent with the examples $S^{+}=\{aab,aaa,ab\}$ and $S^{-}=\{b,aba\}$}
    \label{subfig:example_apta}
    \end{subfigure}
  \end{minipage}
  \hfill
  \begin{minipage}[b]{0.39\textwidth} 
    \centering
    \begin{subfigure}[b]{\textwidth} 
      \centering
         \begin{tikzpicture}[>=stealth, thick, node distance=1.5cm, auto,
         every state/.style={thick, minimum size=2pt}
         ]
          \node[state, initial] (q0) {$q_0$};
          \node[state, accepting] (q1) [right=of q0] {$q_{1}$};
          
          \path[->,>=Stealth]
            (q0) edge [loop above] node {$b$} (q0)
            (q0) edge [above] node {$a$} (q1)
            (q1) edge [loop above] node {$a$,$b$} (q1); 
        \end{tikzpicture}
      \subcaption{Decomposed DFA $\A_{1}$}\label{subfig:dfa1}
    \end{subfigure}
    
    \begin{subfigure}[b]{\textwidth} 
      \centering
      \begin{tikzpicture}[>=stealth, thick, node distance=1.5cm, auto,
      every state/.style={thick, minimum size=2pt}]
          \node[state, initial] (q0) {$q_0$};
          \node[state, accepting] (q1) [right=of q0] {$q_{1}$};
          
          \path[->,>=Stealth]
            (q1) edge [bend left] node {$a$} (q0)
            (q0) edge [bend left] node {$a$,$b$} (q1)
            (q1) edge [loop above] node {$b$} (q1); 
        \end{tikzpicture}
      \subcaption{Decomposed DFA $\A_2$}\label{subfig:dfa2}
    \end{subfigure}
  \end{minipage}
  
  \caption{(a) The generated APTA, where an accepting node is represented by a double circle and a rejecting node is represented by a square. (b) and (c) show the two DFAs in a corresponding $(2,2)$-DFA decomposition $(\A_{1},\A_2)$ consistent with $S = (S^{+},S^{-})$.
  }
  \label{fig:aptasample}
\end{figure}
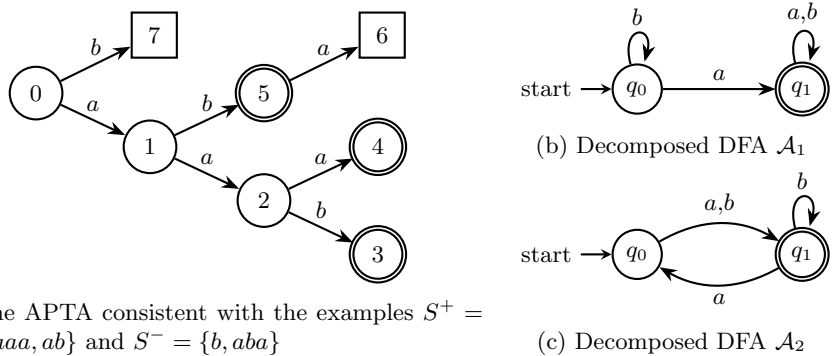


Given a positive integer $n \in \nats$, the state-of-the-art encoding method associates the APTA states with the states of each DFA $\A_{i}$ in the $(m_{1}, \dots, m_{n})$-DFA decomposition $(\A_{1}, \dots, \A_{n})$. 
More precisely, given an APTA consistent with the examples $S$ and an allocation of states $(m_{1}, \dots, m_{n})$, 
each unknown DFA $\A_{i}$ gets assigned $m_{i}$ states and 
the encoding associates each APTA state with one of the $\A_{i}$'s $m_{i}$ states, subject to a set of constraints, such as each accepting (rejecting, resp.) APTA state should be associated with an accepting (a rejecting, resp.) DFA state. 
APTA states with the same (DFA-indexed) state variable will be identified as the same state in the corresponding DFA. 
For instance, for each APTA state $v \in \nats_{S}$, a state variable $x_{v,i}^k$ represents it is associated with the $i$-th state in the $k$-th DFA. 
Considering the APTA in Fig.~\hyperref[subfig:example_apta]{1(a)}, if we would like to generate a $(2, 2)$-DFA decomposition $(\A_{1}, \A_2)$, we need $32$ state variables $x_{v,i}^k$. 
Besides these variables, the encoding method requires other variables to represent the transition relations and the accepting conditions.
Based on the coloring, several constraints imposed by \hyperref[itm:consistency]{\textbf{C1}} -- such as a positive example must be accepted by all DFAs and a negative example must be rejected by at least one DFA -- will be encoded into a SAT problem. 
The complete list of Lauffer's encoding~\cite{LaufferYVSS22} can be found in Appendix~\ref{sec:sat_encoding}. 
Fig.~\hyperref[subfig:dfa1]{1(b)} and~\hyperref[subfig:dfa2]{1(c)} present the identified $(2, 2)$-DFA decomposition $(\A_{1}, \A_2)$ consistent with the APTA in~Fig.~\hyperref[subfig:example_apta]{1(a)}.



\subsection{Our improved encoding via 3DFAs} \label{subsec:3dfa_encoding}
The existing approaches~\cite{HeuleV10,UlyantsevZS15,UlyantsevZS16,ZakirzyanovMIUM19,LaufferYVSS22} face a critical challenge: 
the size of the APTA grows dramatically with both the number and the length of examples $S$ as APTA associates each prefix in $\pref(S)$ with a unique node.
As mentioned earlier, the increase in the number of nodes in the APTA leads to more encoding variables, which in turn increases the size of the resulting SAT problem~\cite{Neider14,ZakirzyanovMIUM19}.
Consequently, the smaller the APTA size is, the easier the SAT problem will be.

In light of this, to obtain smaller APTAs, \cite{DellErbaLS24} proposed to construct the \emph{minimal} 3DFA for $S$ by merging equivalent nodes.
Moreover, inspired by~\cite{DaciukMWW00}, the minimal 3DFA can even be constructed incrementally from the set $S$ if the examples in $S$ are sorted by the standard lexicographical order.
Following~\cite{DaciukMWW00,DellErbaLS24}, we also propose an incremental construction of 3DFAs by merging equivalent nodes.
Differently from~\cite{DellErbaLS24}, our construction, rather than looking for the minimal 3DFA, requires that every rejecting word in $S^{-}$ must reach a unique state in the constructed 3DFA.
That is, we only merge equivalent accepting or ``don't-care'' states.
In what follows, we first introduce our construction of 3DFAs starting from APTAs and then the improved encoding via 3DFAs.
Afterwards, we reveal the reason why it is important to associate every rejecting word with a unique state in the 3DFA. 


\myparagraph{3DFA Construction}
For simplicity, we assume that the full APTA $\mathcal{P}$ of $S$ is given.
Note that our 3DFA can also be constructed on-the-fly from $S$ using the same techniques as~\cite{DaciukMWW00,DellErbaLS24}.
Since our goal is to associate each rejecting word with a unique state and merge as many other states as possible, our reduction process works in a backward manner as follows.

Initially, we collapse all accepting nodes without outgoing transitions into one representative state and store it in a hash map called \texttt{Register}.
Since all rejecting nodes are inequivalent to other states, each rejecting node will be stored as its own representative in the \texttt{Register}.  

Our reduction process then iteratively traverses the APTA nodes in a backward manner from leaves towards the root and process the 3DFA as follows.
In each iteration, we first collect the states whose all successors are representative states in \texttt{Register}, and then identify equivalent states with the following two conditions:
\begin{itemize}
\item 
    both states must be either accepting or don’t-care states, and
\item 
    for every input letter $a \in \Sigma$, either they both have no successors or both have the same successor in \texttt{Register}.
\end{itemize}
We thus create a state representative for each equivalent class, i.e., a set of equivalent states, and store the representative in \texttt{Register}.
Further, all states that have a representative in \texttt{Register} will be replaced by their representative in the updated 3DFA.

Our construction repeats the reduction process until all states, including the initial one, are processed, resulting in a 3DFA consistent with $S$.

\begin{restatable}{theorem}{restateDFACorrectness}\label{thm:dfacorrectness}
    Given a set of examples $S = (S^{+}, S^{-})$, the 3DFA construction produces a 3DFA consistent with $S$.
\end{restatable}

Note that our 3DFA can also be constructed \emph{on the fly} from $S$ in the same manner as in~\cite{DellErbaLS24} if the example words are taken out from $S$ in the standard lexicographical order.

Observe that every rejecting word in $S^{-}$ leads the constructed 3DFA to a unique rejecting state by construction.
In fact, we can get the following stronger result, that we will use later in our SAT encoding phase.
\begin{restatable}{lemma}{restateLemma}\label{lem:prefix-unique}
    Let $\A = (\T, A, R)$ be the outcome of the 3DFA construction, where $\T = (\states, \init, \trans)$.
    For two different prefixes $u, u' \in \pref(S)$, we have $\trans(\init, u) \neq \trans(\init, u')$ if either $u \in \pref(S^{-})$ or $u' \in \pref(S^{-})$.
\end{restatable}



\begin{figure}[t] 
  \begin{minipage}{0.5\textwidth}
  \centering 
  \begin{tikzpicture}[
        node distance=2cm, 
        every state/.style={draw, thick, circle, minimum size=0.7cm}, 
        accepting/.style={double}, 
        false node/.style={draw, thick, rectangle, minimum size=0.6cm} 
      ]
      \node[state] (0) {0};
      \node[state] (1) [below right=0.2cm and 1cm of 0] {1};
      \node[false node] (7) [above right=0.2cm and 1cm of 0] {7};
      \node[state] (2) [below right=0.2cm and 1cm of 1] {2};
      \node[state,accepting] (5) [above right=0.2cm and 1cm of 1] {5};
      \node[state, accepting, fill, color=blue!50] (3) [right= 0.9cm of 2] {\textcolor{black}{3}};
      \node[false node] (6) [above right=0.2cm and 1cm of 5] {6};
      \path[->,>=Stealth,thick]
      (0) edge node [pos=0.5, above] {$a$} (1)
      (0) edge node [pos=0.5, above] {$b$} (7)
      (1) edge node [pos=0.5, above]  {$a$} (2)
      (1) edge node [pos=0.5, above] {$b$} (5)
      (2) edge node [pos=0.5, above]  {$a,b$} (3)
      (5) edge node [pos=0.5, above] {$a$} (6);
    \end{tikzpicture}
  \caption{The 3DFA $\A = (\T, A, R)$ constructed from the APTA in Fig.\hyperref[subfig:example_apta]{1(a)}, where $A=\{3,5\}$ and $R=\{6,7\}.$}\label{fig:3dfa}
  \label{fig:3DFAtree}
  \end{minipage}
  \hfill
  \begin{minipage}{0.45\textwidth}
  \captionof{table}{Size comparison between 3DFAs and APTAs. ``Length'' indicates the length of each word.}
  \label{tab:sizeComparison}
    \centering
    \resizebox{\textwidth}{!}{
    \centering
    \setlength{\tabcolsep}{8pt}
    \begin{tabular}{c c c c}
        \toprule
        \multirow{2}{*}{$|\Sigma|$} & \multirow{2}{*}{Length} & \multicolumn{2}{c}{Automata Size}  \\
        \cmidrule{3-4}
        &  & \textbf{3DFA} & \textbf{APTA}  \\
        \midrule
        5 & 5  & 1,255 & 3,214 \\
        5 & 6  & 4,449 & 13,634 \\
        5 & 7  & 13,787 & 53,277 \\
        5 & 8  & 43,064 & 209,721 \\
        5 & 9  & 136,019 & 835,954 \\
        5 & 10 & 443,763 & 3,369,694 \\
        \bottomrule
    \end{tabular}
    }
  \end{minipage}
\end{figure}

For instance, as shown in Fig.~\ref{fig:3DFAtree}, this construction process can merge the states $\{3,4\}$ of the APTA in Fig.~\hyperref[subfig:example_apta]{1(a)} reached by positive examples into a representative state $3$. 
Although this simple example does not clearly show the advantage of the construction, Table~\ref{tab:sizeComparison} presents the size comparison between APTAs and our constructed 3DFAs in terms of the number of states on a number of cases from parity game solving~\cite{DellErbaLS24}. 
We can see that the resulting 3DFAs exhibit significantly smaller sizes than the original APTAs. 
Therefore, using 3DFAs instead of APTAs requires dramatically fewer variables and fewer constraints in our improved encoding given below, yielding easier SAT problems and thus faster solving speed in general than those via APTAs, as confirmed by the experiments in Section~\ref{sec:experiments}.





\myparagraph{SAT Encoding via 3DFAs}
We now present the standard SAT encoding for solving the DFA-DIP problem~\cite{LaufferYVSS22} adapted to our 3DFAs. 
Note that we also use the symmetry breaking techniques proposed in~\cite{UlyantsevZS15,UlyantsevZS16} to further improve our SAT encoding, just like in~\cite{LaufferYVSS22}. In what follows, we focus on our major encoding. 
Given the 3DFA $\A = (\T, A, R)$ with $\T = (\states, \init, \trans)$ consistent with $S$, we are looking for a consistent DFA decomposition $(\A_{1}, \dots, \A_{n})$ of $S$ where the state allocation is $(m_{1}, \dots, m_{n})$ for a given $n \in \nats$.
To encode the DFA-DIP, we use the following three types of variables:
\begin{enumerate}
\item 
    \textbf{State variables} $x_{v,i}^{k}$, where $k \in [n]$, $v \in \states$, and $i \in [m_{k}]$. 
    $x_{v,i}^{k} \equiv 1$ iff the state $v$ of the 3DFA and the state $i$ in the DFA $\A_{k}$ can both be reached on some word $u \in \Sigma^{*}$ from their initial states.
\item 
    \textbf{Transition relation variables} $e^{k}_{l,i,j}$, where $ k \in [n]$, $l \in \Sigma$, and $ i,j \in [m_{k}]$. 
    $e^{k}_{l,i,j} \equiv 1$ iff DFA $\A_{k}$ has a transition from state $i$ to state $j$ over the letter $l$.
\item 
    \textbf{Acceptance variables} $z_{i}^{k}$, where $k \in [n]$ and $i \in [m_{k}]$. 
    $z_{i}^{k} \equiv 1$ iff the state $i$ of DFA $\A_{k}$ is an accepting state.
\end{enumerate}

Recall that $A$ (resp., $R$) is the set of accepting (resp., rejecting) states in the 3DFA $\A$.
We now give the list of constraints for the SAT encoding.
First of all, we require that each individual automaton $\A_{i}$ should be a DFA.
That is, $\A_{i}$ must be deterministic and also complete.
\begin{description}
\item[\namedlabel{constraint:D1}{D1 Determinism.}] 
    For a state $i$ and a letter $l$  in $\A_{k}$, there is at most one successor:
    \[
        \bigwedge_{l \in \Sigma} \bigwedge_{k \in [n]} \bigwedge_{\substack{i,j,t \in [m_{k}] \\ j < t}} e^{k}_{l,i,j} \implies \neg e^{k}_{l,i,t}.
    \]
\item[\namedlabel{constraint:D2}{D2 Completeness.}] 
    For a state $i$ and a letter $l$ in $\A_{k}$, there must be a successor: 
    \[
        \bigwedge_{l \in \Sigma} \bigwedge_{k \in [n]} \bigwedge_{i \in [m_{k}]} \bigvee_{j \in [m_{k}]} e^{k}_{l,i,j}.
    \]
\end{description}

Second, we require that the DFA decomposition $(\A_{1}, \cdots, \A_{n})$ should be consistent with $S$, so to satisfy \hyperref[itm:consistency]{\textbf{C1}}.
\begin{description}
\item[\namedlabel{constraint:R1}{R1 Positive Consistency.}] 
    Every positive example in $S^{+}$ must be accepted by all individual DFAs, i.e., by each DFA $\A_{k}$ where $k \in [n]$. It follows that, if $x^k_{v,i} \equiv 1$, then all positive examples leading the 3DFA $\A$ to an accepting state $v$ must also make $\A_{k}$ reach an accepting state, i.e., state $i$ must also be an accepting state.
    \[
        \bigwedge_{v \in A} \bigwedge_{k \in [n]} \bigwedge_{i \in [m_{k}]} x^{k}_{v,i} \implies z^{k}_{i}.
    \]
\item[\namedlabel{constraint:R2}{R2 Negative Consistency.}] 
    Each negative example in $S^{-}$ should be rejected by at least one individual DFA $\A_{k}$. That is, if $x^k_{v,i} \equiv 1$, then the example leading the 3DFA $\A$ to a rejecting state $v$ must also make $\A_{k}$ reach a rejecting state as well, i.e., state $i$ of $\A_{k}$ should also be rejecting.
    \[
        \bigwedge_{v \in R} \bigvee_{k \in [n]} \big(\bigwedge_{i \in [m_{k}]} \left( x_{v,i}^{k} \implies \neg z_{i}^{k} \right)\big).
    \]
\end{description}

To further enforce the DFA decomposition to be consistent with $S$, we also need to perform the product of the 3DFA $\A$ and each individual DFA $\A_{k}$ in order to build the transition system of each $\A_{k}$, where $k \in [n]$.
That is, $x^k_{v,i} \equiv 1$ also means that in the product automaton of $\A$ and $\A_{k}$, the pair of states $(v, i)$ is seen as a product state and it is reached from the initial product state $(r, 1)$ over some word $u \in \Sigma^{*}$. 
The constraints are listed below:
\begin{description}
\item[\namedlabel{constraint:T1}{T1 Initialization.}] 
    The initial state $r$ of 3DFA $\A$ should always be associated with the initial state $1$ for each DFA $\A_{k}$: 

    \[
        \bigwedge_{k \in [n]} x^k_{r, 1}
    \]
\item[\namedlabel{constraint:T2}{T2 State Correspondence.}] 
    Each state $v$ of 3DFA $\A$ must be associated with one state in each DFA $\A_{k}$, i.e., $x^k_{v,i} \equiv 1$ for some state $i$ in $\A_{k}$: 
    \[
        \bigwedge_{v \in Q} \bigwedge_{k \in [n]} \bigvee_{i \in [m_{k}]} x^{k}_{v,i}.
    \]
\item[\namedlabel{constraint:T3}{T3 Transition Relation.}]
     For each $k\in [n]$, in the product automaton of $\A$ and $\A_{k}$, if the product state $(v, i)$ is reachable from the initial product state $(r, 1)$, and there is a transition from state $i$ to $j$ over letter $a$ in $\A_{k}$, then the product state $(\delta(v, a), j)$ is also reachable. 
    \[
        \bigwedge_{v \in Q} \bigwedge_{k \in [n]} \bigwedge_{i,j \in [m_{k}]} \bigwedge_{a \in l(v)} (x^{k}_{v,i} \land e^{k}_{a,i,j}) \implies x^{k}_{\delta(v, a),j} 
    \]
    where $l(v)$ is the set of letters on the outgoing transitions from state $v$.
\end{description}


In addition to the above constraints, Lauffer \emph{et al.}~\cite{LaufferYVSS22} also proposed an optimization as follows. 
\begin{description}
\item[\namedlabel{constraint:O1}{O1}] 
    Each state $v$ in the 3DFA can only be associated with at most one state $i$ in each individual DFA $\A_{k}$. 
    \[
         \bigwedge_{v \in Q} \bigwedge_{k \in [n]} \bigwedge_{i,j \in [m_{k}], i < j} ( \neg x^{k}_{v,i} \lor \neg x^{k}_{v,j}). 
    \]
\end{description}
This constraint holds in the APTA because each state $v$ corresponds to a unique prefix $u \in \pref(S)$.
Since both $\A$ and $\A_{k}$ are deterministic, the product state $(v, i)$ reached from the initial product state $(r, 1)$ over $u$ must be unique.
That is, there will be a unique $i$ in $\A_{k}$ to make $x^k_{v,i}$ be true.

However, this constraint might not be true when we use a 3DFA for our encoding.
For instance, consider the 3DFA $\A$ in Fig.~\ref{fig:3dfa}:
we can see that there are two words that can make $\A$ reach state $3$, namely $aaa$ and $aab$.
In our target DFA $\A_{k}$, these two words might lead to two different states, say $i$ and $j$.
It then follows that we have both $x^k_{3, i} \equiv 1$ and $x^k_{3, j} \equiv 1$, which obviously violates the constraint \hyperref[constraint:O1]{\textbf{O1}}.
This will prevent us from finding a consistent DFA decomposition $(\A_{1}, \dots, \A_{n})$ with the state allocation $(m_{1}, \dots, m_{n})$.

Nonetheless, we observe that during the construction of the 3DFA, it is easy to identify those states that correspond to multiple different prefix words in $\pref(S)$.
Indeed, we can just record the representative states which correspond to multiple equivalent states, denoted by $M$.
That is, for each state $v \in M$, there must be two different prefixes $u, u' \in \pref(S)$ such that $v = \delta(r, u) = \delta(r, u')$.
In fact, $M$ must not contain the states corresponding to a prefix in $\pref(S^-)$ as guaranteed by Lemma~\ref{lem:prefix-unique}.
The set of states in $\A$ that is associated with a unique prefix word in $\pref(S)$ is $Q\setminus M$.
Therefore, we can replace the constraint \hyperref[constraint:O1]{\textbf{O1}} with the following constraint:
\begin{description}
\item[\namedlabel{constraint:O1prime}{O1'}] 
    Each state $v$ in the 3DFA that does not correspond to multiple equivalent states can only be associated with at most one state $i$ in each individual DFA $\A_{k}$.
    \[
        \bigwedge_{v \in Q\setminus M} \bigwedge_{k \in [n]} \bigwedge_{i,j \in [m_{k}], i < j} (\neg x^{k}_{v,i} \lor \neg x^{k}_{v,j}).
    \]
\end{description}
Clearly, by Lemma~\ref{lem:prefix-unique}, all states having a path to a rejecting state in $\A$ belong to $Q \setminus M$. 

Let $\varphi^{\A, n}_{(m_{1}, \dots, m_{n})}$ be the conjunction of the constraints \hyperref[constraint:D1]{\textbf{D1}}-\hyperref[constraint:D2]{\textbf{2}}, \hyperref[constraint:R1]{\textbf{R1}}-\hyperref[constraint:R2]{\textbf{2}}, \hyperref[constraint:T1]{\textbf{T1}}-\hyperref[constraint:T3]{\textbf{3}}, and \hyperref[constraint:O1prime]{\textbf{O1'}}.
Then, we get the following result.
\begin{restatable}{theorem}{restateEncodingCorrectness}\label{thm:encoding_correctness}
    Let $\A$ be the 3DFA consistent with the examples $S$ and $n \in \nats$. 
    $\varphi^{\A, n}_{(m_{1}, \dots, m_{n})}$ is satisfiable if, and only if, there exists a $(m_1,\dots,m_n)$-DFA decomposition $(\A_{1}, \dots, \A_{n})$ consistent with $S$.
\end{restatable}

\myparagraph{Encoding size}
The size of the formula $\varphi^{\A, n}_{(m_{1}, \dots, m_{n})}$ is determined by the conjunction of each constraint above. 
Let $m$ be the maximal number in $(m_{1}, \dots, m_{n})$. 
The contributions of each constraint to the formula size are as follows: 
(\hyperref[constraint:D1]{\textbf{D1}}-\hyperref[constraint:D2]{\textbf{2}}): $O(|\Sigma| \cdot n \cdot m^{3}) + O(|\Sigma| \cdot n \cdot m^{2})$; 
(\hyperref[constraint:R1]{\textbf{R1}}-\hyperref[constraint:R2]{\textbf{2}}): $O(|Q^{+}| \cdot n \cdot m) + O(|Q^{-}| \cdot n \cdot m)$; 
(\hyperref[constraint:T1]{\textbf{T1}}-\hyperref[constraint:T3]{\textbf{3}}): $O(n) + O(|Q| \cdot n \cdot m) + O(|Q| \cdot n \cdot m^{2})$;
and 
(\hyperref[constraint:O1prime]{\textbf{O1'}}): $O(|Q \setminus M| \cdot n \cdot m^{2})$.
The overall size is thus $O(|\Sigma| \cdot n \cdot m^{3}) + O(|Q| \cdot n \cdot m^{2})$. 
$\hfill \lhd$




\subsection{Why not use minimal 3DFAs} \label{subsec:encoding_problem_negative}

Our 3DFA construction is inspired by the one of \texttt{DFAMiner}~\cite{DellErbaLS24}, with the main difference being that we do not generate the minimal 3DFA recognizing $S$.
In the minimal 3DFA generated by \texttt{DFAMiner}, all rejecting states that cannot reach a rejecting state are merged together.
For instance, in Fig.~\ref{fig:3dfa}, both rejecting states $6$ and $7$ do not reach another rejecting state, so in the minimal 3DFA, these two states will be merged.
This also means that the two rejecting words $b$ and $aba$ will correspond to the same state, say $v$, in the minimal 3DFA $\A$.
However, this would make the meaning of constraint~\hyperref[constraint:R2]{\textbf{R2}} imprecise:
by constraint~\hyperref[constraint:R2]{\textbf{R2}}, we want every word to be rejected by at least one individual DFA, but the constraint actually says that we have at least one individual DFA that rejects a word associated with the rejecting state $v$.
This means that it is \emph{not} guaranteed that every word in $S^{-}$ reaching $v$ will be rejected.
This means, for instance, that for $b$ and $aba$ in Fig.~\ref{fig:3dfa}, it may happen that only $b$ is rejected by some individual DFA but $aba$ is totally ignored and thus not rejected by any DFA.
This kind of issues can lead to inconsistent DFA decompositions from the SAT solver and we indeed observed them in the experiments when we directly used the minimal 3DFA in our encoding.

In light of this, it is clear why we require in our construction that every rejecting word must be associated with a unique state in the constructed 3DFA or, equivalently, that every rejecting state in the APTA must not be equivalent to any other states.
We are now ready to introduce in the next section the overall algorithm for finding consistent DFA decompositions.

\section{Identification Algorithms for DFA Decomposition}
\label{sec:overall-algo}

In this section, we present identification algorithms for two DFA decomposition problems: 
the Pareto-optimal DIP and the states-optimal DIP. 

For the Pareto-optimal DIP, our algorithm builds on top of the approach of~\cite{LaufferYVSS22}, where we replace their SAT encoding via APTA with our enhanced 3DFA-based encoding. 
The detailed algorithm can be found in Appendix~\ref{sec:sat_encoding}. 
As the experimental results in Section~\ref{sec:experiments} show, this modification yields significant improvement. 
Therefore, in what follows, we focus on our method for solving the new states-optimal DIP, which is summarized as Algorithm~\ref{alg:states_optimal_DIP}.


\begin{algorithm}[t]
    \caption{State-optimal DIP Solving}
    \label{alg:states_optimal_DIP}
    \SetKwInput{Input}{Input}\SetKwInOut{Output}{Output}\SetNoFillComment
    \Input{The labeled examples $S=\{S^{+},S^{-}\}$}
    \Output{A DFA decomposition $\mathcal{D}$ for the states-optimal DIP}
        $\A \gets \proc{3DFAConstruction}(S)$\tcp*{Construct a 3DFA from the examples} \label{alg1_line:3DFA}
        $N \gets 2$\tcp*{Initial total number of states in possible decompositions} \label{alg1_line:initial}
        \While{\True}{
            $\mathcal{M} \gets \proc{ComputeStatesAllocations}(N, 2)$\tcp*{Get all possible states allocations under $N$} \label{alg1_line:computeCombinations}
            \ForEach{$(m_{1}, \dots, m_{n}) \in \mathcal{M}$}{ 
                $\mathit{SAT}, \mathcal{D} \gets \proc{Solve}((m_{1}, \dots, m_{n}), \A)$\; \label{alg1_line:solve}
                \If{$\mathit{SAT}$}{
                    \Return $\mathcal{D}$\; \label{alg1_line:solution}
                }
            }
            $N \gets N+1$\; \label{alg1_line:increase_{n}}
        }
\end{algorithm}

Given a set of labeled examples $S = (S^{+}, S^{-})$, Algorithm~\ref{alg:states_optimal_DIP} first builds a 3DFA consistent with $S$ (Line~\ref{alg1_line:3DFA}) according to Section~\ref{subsec:3dfa_encoding} and then looks for the decomposition with the minimal total number of states, starting with $N = 2$ (Line~\ref{alg1_line:initial}) and incrementing it until a suitable decomposition is found. 
This is achieved by computing all possible states allocations $\mathcal{M}$ (Line~\ref{alg1_line:computeCombinations}) having $N$ total states by calling Algorithm~\ref{alg:computeDFAcombinations}, which ensures that every states allocation $(m_{1}, \dots, m_{n}) \in \mathcal{M}$ satisfies: 
(1) the total number of states is $N$ (i.e., $\sum_{i=1}^{n} m_{i} = N$),  
(2) every corresponding DFA has at least $2$ states\footnote{As the constraint \hyperref[constraint:D2]{\textbf{D2}} in Section~\ref{subsec:3dfa_encoding} requires that every state in the generated DFAs should have a transition for every action in $\Sigma$, the generated DFAs are complete DFAs. Since a complete DFA with a single state will either accept or reject all words, we skip such na\"ive DFAs and thus every DFA should have at least 2 states.} (i.e., $m_{i} \geq 2$ for all $1 \leq i \leq n$), and
(3) allocated states are in ascending order (i.e., $m_{i} \leq m_{i+1}$ for all $1 \leq i < n$). 
As a result, Algorithm~\ref{alg:computeDFAcombinations} enumerates all possible combinations where the number of DFAs $n$ will range from $1$ to $\left\lfloor \frac{N}{2} \right\rfloor$. 
For instance, for $N = 10$, Algorithm~\ref{alg:computeDFAcombinations} returns the following states allocations:
\vspace{-0.2cm}
\[
\begin{array}{c}
\{
(2,2,2,2,2), (2,2,3,3), (2,2,2,4), (3,3,4), (2,4,4), (2,3,5), \\ (2,2,6), (5,5), (4,6), (3,7), (2,8), (10)
\}.
\end{array}
\vspace{-0.2cm}
\]
After calling Algorithm~\ref{alg:computeDFAcombinations}, for every possible states allocation $(m_{1}, \dots, m_{n}) \in \mathcal{M}$, Algorithm~\ref{alg:states_optimal_DIP} applies our encoding method and calls a SAT solver to determine whether there is a $(m_{1}, \dots, m_{n})$-DFA decomposition $\mathcal{D}$ for the state-optimal DIP (Line~\ref{alg1_line:solve}). 
If so, it returns it (Line~\ref{alg1_line:solution}). 
Otherwise, it means that there is no solution under the current total number of states $N$, so we increase it (Line~\ref{alg1_line:increase_{n}}). 

Note that, in Algorithm~\ref{alg:states_optimal_DIP}, we increase the total number of states $N$ by 1 in each round. 
For every fixed $N$, Algorithm~\ref{alg:computeDFAcombinations} sorts the states allocations in descending order based on their entropy values (Line~\ref{alg2_line:sorted_by_entropy}). 
Therefore, they work together to keep the states-optimal preorder defined in Section~\ref{sec:problems}. Theorem~\ref{thm:algo_termination_and_correctness} shows the termination and the correctness of Algorithm~\ref{alg:states_optimal_DIP}. 



\begin{algorithm}[t]
    \caption{\proc{ComputeStatesAllocations}}
    \label{alg:computeDFAcombinations}
    \SetKwInput{Input}{Input}\SetKwInOut{Output}{Output}\SetNoFillComment
    \Input{$N$: the total number of states in unknown decompositions \\ $\qquad\quad\;\,$ $k$: the minimum number of states in each DFA}
    \Output{$\mathcal{M}$: all possible states allocations of unknown decompositions}
        $\mathcal{M} \gets \{N\}$\;
        \For(\tcp*[f]{Current DFA has $m$ states}){$m \gets k$ \KwTo $N$}{
            \If(\tcp*[f]{Still enough states left for the next DFA}){$m  \leq N - m$}{
                $\mathcal{M}' \gets \proc{ComputeStatesAllocations}(N - m, m)$\;
                $\mathcal{M} \gets \mathcal{M} \cup (\{m\} \times \mathcal{M}')$\tcp*{Add the new states allocations to $\mathcal{M}$}
            }
        }
        Sort $\mathcal{M}$ in descending order by entropy (see Definition~\ref{def:entropy})\; \label{alg2_line:sorted_by_entropy}
        \Return $\mathcal{M}$\;
\end{algorithm}

\begin{restatable}[Termination and Correctness]{theorem}{restateTerminationCorrectness} \label{thm:algo_termination_and_correctness}
Let $S = (S^{+}, S^{-})$ be a given set of labeled examples. 
Algorithm~\ref{alg:states_optimal_DIP} terminates and returns a correct DFA decomposition $\mathcal{D} = (\A_{1}, \dots, \A_{n})$ for the states-optimal DIP.
\end{restatable}



\myparagraph{Complexity analysis} 
It can be observed that there is always a correct decomposition $\mathcal{D} = (\A_{1}, \dots, \A_{n})$ when $N = 2 + \sum_{u \in S^{-}}{(|u|+2)}$, where a 2-states DFA $\A_{1}$ accepts all words and every other DFA rejects one negative example $u \in S^{-}$.
Therefore, the loop iterations in Algorithm~\ref{alg:states_optimal_DIP} are bounded by $1 + \sum_{u \in S^{-}} {(|u|+2)}$, as $N$ starts with $2$ and it is increased by $1$ every round. 
In every loop, the recursive generation of all possible states allocations under a given $N$ in Algorithm \ref{alg:computeDFAcombinations} is dominated by the integer partition problem. According to the Hardy-Ramanujan formula\cite{Hardy}, the asymptotic estimation of the number of states allocations is $|\mathcal{M}| \sim \frac{e^{\pi\sqrt{2N/3}}}{4N\sqrt{3}} $, which shows that the growth of $|\mathcal{M}|$ is subexponential in $N$. $\hfill \lhd$





It is not hard to have a simple variant of Algorithm~\ref{alg:states_optimal_DIP} with a given integer $n$ restricting the number of DFAs in the decomposition. 
It can be done by checking if a combination has at most $n$ DFAs before adding it to $\mathcal{M}$ in Algorithm~\ref{alg:computeDFAcombinations}.

\section{Experimental Evaluation}
\label{sec:experiments}

We have implemented our approach\footnote{Available at: \url{https://github.com/MJJ-Shuai/dfa_decomposed_3DFA}} on top of the tool developed in~\cite{LaufferYVSS22} and we used as benchmarks pairs of sets of examples representing partially-ordered tasks (cf.~\cite[Section~III.C]{LaufferYVSS22}), where we vary the numbers of tasks, the maximum length of each sequence of tasks, and the number of samples in each positive and negative sets; 
for each combination, we randomly generate 10 instances, provided they can be generated,  by following the generation strategy\footnote{See \url{https://github.com/mvcisback/dfa-identify}} in~\cite{LaufferYVSS22}.
Let \toolOriginal, \toolLaufferOurEncoding, and \toolOurStateOptimalDIP denote the original tool from~\cite{LaufferYVSS22}, its version using our 3DFA encoding (cf.~Section~\ref{sec:encoding}), and by our method for solving states-optimal DIP (cf.~Section~\ref{sec:overall-algo}), respectively.
We ran the tools on a desktop machine with an i7-4790 CPU and 16 GB of memory running Ubuntu Server 24.04.2 and we used \benchexec~\cite{DBLP:journals/sttt/BeyerLW19} to trace and constrain the tools' executions: 
we allowed each benchmark to use 15 GB of memory and imposed a time limit of 10 minutes of wall-clock time.

\subsection{Comparison between \toolOriginal and \toolLaufferOurEncoding on solving Pareto-optimal DIP}

\begin{table}[t]
    \caption{Overview of the outcomes of the encoding experiments}
    \label{tab:expOverviewPareto}
    \centering
    \setlength{\tabcolsep}{2.2pt}
    \begin{tabular}{c|l|rrrrr|rrrrr}
    \toprule
        \multirow{2}{*}{Tool} & \multicolumn{1}{c|}{\multirow{2}{*}{Result}} & \multicolumn{5}{c|}{\#~DFAs ($|\Sigma| = 2)$} & \multicolumn{5}{c}{\#~DFA ($|\Sigma| = 4$)} \\
        & & \multicolumn{1}{c}{2} & \multicolumn{1}{c}{3} & \multicolumn{1}{c}{4} & \multicolumn{1}{c}{5} & \multicolumn{1}{c|}{6} & \multicolumn{1}{c}{2} & \multicolumn{1}{c}{3} & \multicolumn{1}{c}{4} & \multicolumn{1}{c}{5} & \multicolumn{1}{c}{6}\\
        \hline
        \multirow{3}{*}{\toolOriginal~\cite{LaufferYVSS22}} & Success & 460 & 460 & 460 & 456 & 92 & 1058 & 129 & 85 & 72 & 72 \\
        & Memoryout & 0 & 0 & 0 & 4 & 397 & 21 & 467 & 981 & 989 & 978 \\
        & Timeout & 0 & 0 & 0 & 0 & 1 & 11 & 3 & 24 & 29 & 40 \\
        \hline
        \multirow{2}{*}{\toolLaufferOurEncoding~(ours)} & Success & 460 & 460 & 460 & 460 & 460 & 1090 & 1090 & 1090 & 1090 & 1089 \\
        & Timeout & 0 & 0 & 0 & 0 & 0 & 0 & 0 & 0 & 0 & 1 \\
    \bottomrule
    \end{tabular}
\end{table}

In Table~\ref{tab:expOverviewPareto} we report on how \toolOriginal and \toolLaufferOurEncoding performed on 460 benchmarks with 2 tasks ($|\Sigma| = 2$) and 1090 benchmarks with 4 tasks ($|\Sigma| = 4$) when changing the number of DFAs in the decomposition.
As we can see from the table, \toolOriginal can scale to 5 DFAs for the $|\Sigma| = 2$ benchmarks, but it already struggles at 3 DFAs when $|\Sigma| = 4$.
By just replacing the original APTA encoding with our 3DFA-based one, \toolLaufferOurEncoding has been able to solve all benchmarks except for one case up to 6 DFAs.
We have also run \toolLaufferOurEncoding up to 10 DFAs, obtaining success everywhere except for 2 and 3 timeouts on $|\Sigma| = 4$ for 9 and 10 DFAs, respectively.

\begin{figure}[t]
    \includegraphics{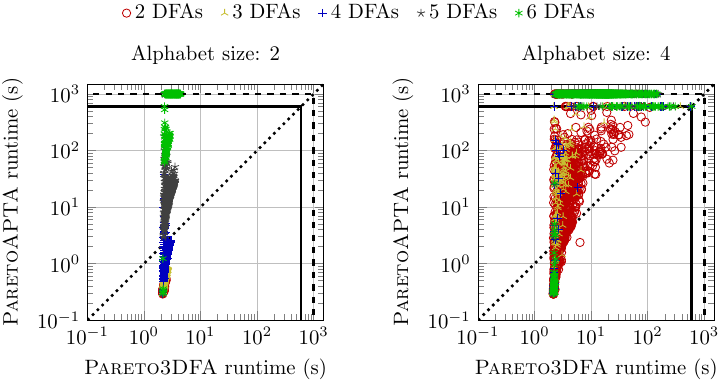}
    \caption{Running time comparison between \toolOriginal and \toolLaufferOurEncoding}
    \label{fig:expScatterLauffer}
\end{figure}
In Fig.~\ref{fig:expScatterLauffer} we show the runtime comparison between two methods \toolOriginal and \toolLaufferOurEncoding on all 1550 experiments considered in Table~\ref{tab:expOverviewPareto}. 
The scatter plots in the figure have logarithmic axes and marks above the dotted diagonal line mean that \toolOriginal took more time than \toolLaufferOurEncoding to solve the same benchmark;
the solid line at 600 seconds represents the timeout we imposed to the experiments while marks on the dashed line at 1000 seconds stand for to experiments where the corresponding tool went memoryout. 

As we can see from the plots, except for the cases taking very limited time, our \toolLaufferOurEncoding always significantly outperforms \toolOriginal in the running time, while producing DFAs with the same number of states as \toolOriginal on the commonly solved cases.
For benchmarks requiring at least 5 seconds to be computed by both tools, \toolLaufferOurEncoding is 1.4--80.3 times faster than \toolOriginal; for at least 10 seconds, the speedup lies in 2.8--40.5.

More detailed comparison on runtime can be found in Appendix~\ref{appendix:experiments}.

\subsection{The scalability for our method on solving states-optimal DIP}
\label{subsec:experimentsSODIP}

\begin{figure}
    \centering
    \includegraphics[scale=1.1]{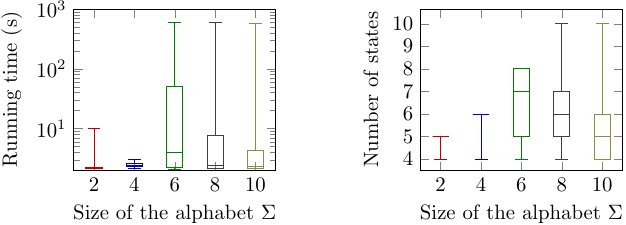}
    \caption{Box plots for the \toolOurStateOptimalDIP experiments}
    \label{fig:expBoxesSODIP}
\end{figure}

To evaluate how \toolOurStateOptimalDIP is able to scale with more challenging benchmarks, we ran it on the same benchmarks used for Table~\ref{tab:expOverviewPareto} as well as on 1150, 1190, and 1210 benchmarks with $|\Sigma| = 6$, $|\Sigma| = 8$, and  $|\Sigma| = 10$, respectively.
\toolOurStateOptimalDIP has completed successfully all experiments for $|\Sigma| \in \setnocond{2, 4}$; 
for $|\Sigma| = 6$ it solved 919 cases and went timeout on 231, while for $|\Sigma| = 8$ it solved 689 cases and went timeout on 501, and for $|\Sigma| = 10$ it solved 783 cases and went timeout on 427;
no failure by memoryout happened.

The box plots in Fig.~\ref{fig:expBoxesSODIP} show the distribution of the running time and of the number of states relative to the successfully solved benchmarks.
As we can see from the plots, \toolOurStateOptimalDIP is really fast for the simpler benchmarks with $|\Sigma| \in \setnocond{2, 4}$, taking less than 1 second and needing between 4 and 6 states to solve each of them.
For the more demanding benchmarks ($|\Sigma| \in \setnocond{6, 8, 10}$), more states are necessary (at most 10), so by Algorithm~\ref{alg:states_optimal_DIP} more cycles in the loop, more decompositions, and larger encoding formulas are generated and evaluated, as reflected by the higher running times shown in the plot on the left of Fig.~\ref{fig:expBoxesSODIP}. 
Appendix~\ref{appendix:experiments} presents more analysis on the behavior of \toolOurStateOptimalDIP.

\section{Conclusion and Future Work} \label{sec:conclusion}

In this paper we considered two DFA decomposition identification problems: the Pareto-optimal DIP, studied by Lauffer \emph{et al.} in~\cite{LaufferYVSS22}, and the states-optimal DIP that we introduced in this paper. 
To solve the former problem, we proposed an improved SAT encoding via 3DFA; 
compared to the encoding via APTA~\cite{LaufferYVSS22}, our method reduces the number of required encoding variables, thus significantly improving the efficiency, as confirmed by the experimental results, showing that our method is dramatically faster than the state-of-the-art method from~\cite{LaufferYVSS22}. 
We also proposed a solution method for the novel states-optimal DIP, and the experimental results on a large set of benchmarks demonstrate its scalability. 

For future work, we consider further improving the practical efficiency in running time. 
One possibility is to find a novel encoding method where the input is a fixed total number of states $N$ of the decomposition, but not a states allocation. In this way, if the answer is UNSAT, we can just increase the total number of states $N$ without computing all possible states allocations suitable for the current $N$.



\subsubsection{\ackname}
We thank the anonymous reviewers for their useful remarks that helped us improve the quality of the paper.
Work supported in part by the National Key R\&D Program of China under grant No.~2022YFA1005101, the Natural Science Foundation of China (NSFC) under grants No.~W2511064, No.~62472316, No.~62192732, No.~62032024 and No.~62032019, and the ISCAS Basic Research Grant No.~ISCAS-JCZD202406.
\newline\protect\includegraphics[height=8pt]{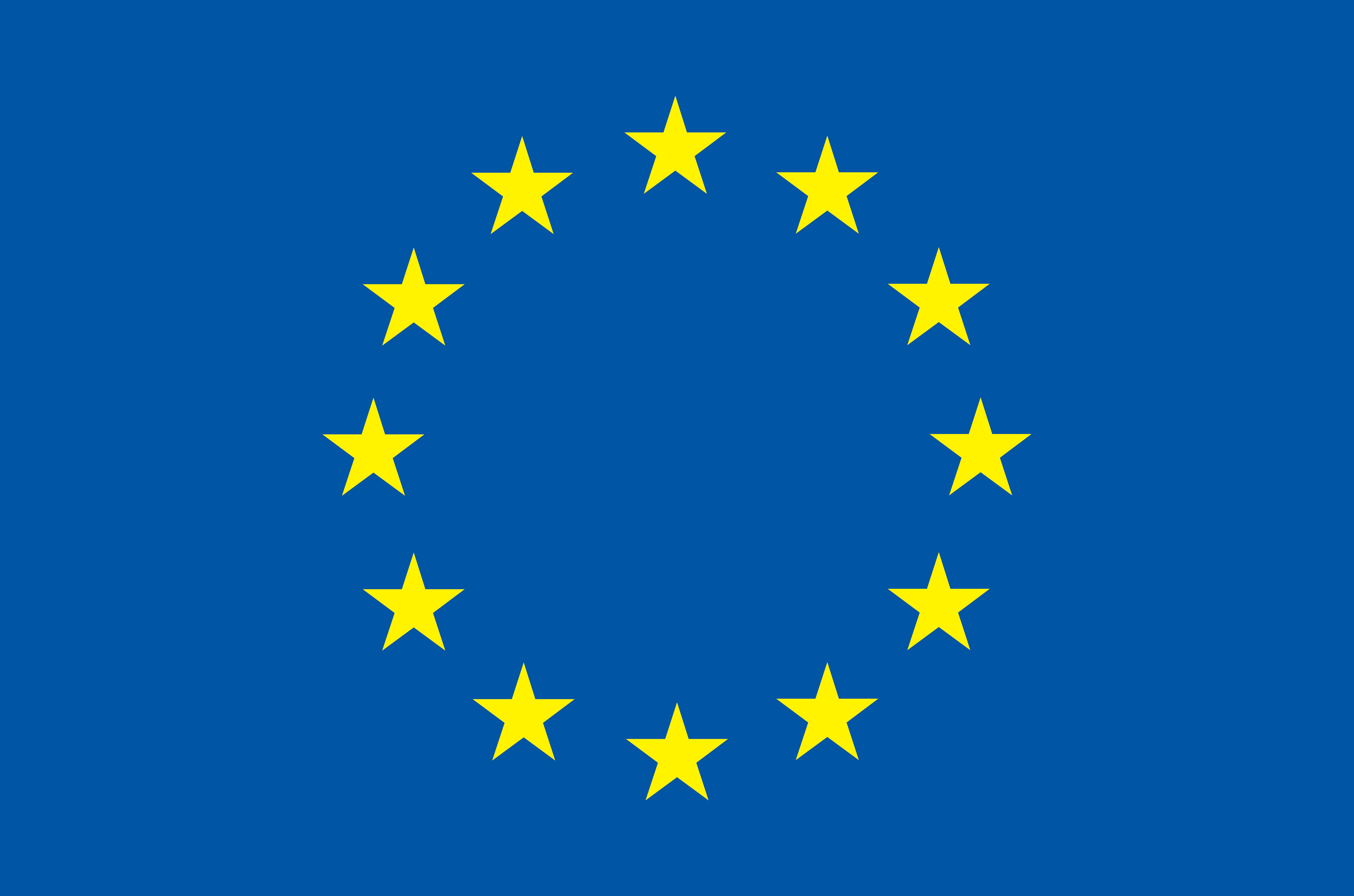} This project is part of the European Union’s Horizon 2020 research and innovation programme under the Marie Sk\l{}odowska-Curie grant no.\@ 101008233.

%
%
%
\bibliographystyle{splncs04}
\bibliography{reference}

\begin{thebibliography}{10}
\providecommand{\url}[1]{\texttt{#1}}
\providecommand{\urlprefix}{URL }
\providecommand{\doi}[1]{https://doi.org/#1}

\bibitem{Ashar1992}
Ashar, P., Devadas, S., Newton, A.R.: Finite State Machine Decomposition, pp.
  117--168. Springer US, Boston, MA (1992)

\bibitem{DBLP:journals/sttt/BeyerLW19}
Beyer, D., L{\"{o}}we, S., Wendler, P.: Reliable benchmarking: requirements and
  solutions. Int. J. Softw. Tools Technol. Transf.  \textbf{21}(1),  1--29
  (2019)

\bibitem{CosteN98}
Coste, F., Nicolas, J.: How considering incompatible state mergings may reduce
  the {DFA} induction search tree. In: {ICGI} 1998. Lecture Notes in Computer
  Science, vol.~1433, pp. 199--210. Springer (1998)

\bibitem{DaciukMWW00}
Daciuk, J., Mihov, S., Watson, B.W., Watson, R.E.: Incremental construction of
  minimal acyclic finite state automata. Comput. Linguistics  \textbf{26}(1),
  3--16 (2000)

\bibitem{DellErbaLS24}
Dell'Erba, D., Li, Y., Schewe, S.: {DFAMiner}: Mining minimal separating {DFA}s
  from labelled samples. In: {FM} 2024 {(2)}. Lecture Notes in Computer
  Science, vol. 14934, pp. 48--66. Springer (2024)

\bibitem{Egri-Nagy11}
Egri{-}Nagy, A.: Applications of automata theory and algebra via the
  mathematical theory of complexity to biology, physics, psychology,
  philosophy, and games. Artif. Life  \textbf{17}(2),  141--143 (2011)

\bibitem{DBLP:conf/cade/GrinchteinLP06}
Grinchtein, O., Leucker, M., Piterman, N.: Inferring network invariants
  automatically. In: {IJCAR} 2006. Lecture Notes in Computer Science,
  vol.~4130, pp. 483--497. Springer (2006)

\bibitem{Hardy}
Hardy, G.H., Ramanujan, S.: Asymptotic formulae in combinatory analysis.
  Proceedings of the London Mathematical Society  \textbf{s2-17}(1),  75--115
  (1918)

\bibitem{HeuleV10}
Heule, M., Verwer, S.: Exact {DFA} identification using {SAT} solvers. In:
  {ICGI} 2010. Lecture Notes in Computer Science, vol.~6339, pp. 66--79.
  Springer (2010)

\bibitem{DBLP:journals/pr/Higuera05}
de~la Higuera, C.: A bibliographical study of grammatical inference. Pattern
  Recognit.  \textbf{38}(9),  1332--1348 (2005)

\bibitem{KupfermanM15}
Kupferman, O., Mosheiff, J.: Prime languages. Inf. Comput.  \textbf{240},
  90--107 (2015)

\bibitem{DBLP:conf/icgi/LangPP98}
Lang, K.J., Pearlmutter, B.A., Price, R.A.: Results of the abbadingo one {DFA}
  learning competition and a new evidence-driven state merging algorithm. In:
  {ICGI} 1998. Lecture Notes in Computer Science, vol.~1433, pp. 1--12.
  Springer (1998)

\bibitem{LaufferYVSS22}
Lauffer, N., Yalcinkaya, B., Vazquez{-}Chanlatte, M., Shah, A., Seshia, S.A.:
  Learning deterministic finite automata decompositions from examples and
  demonstrations. In: {FMCAD} 2022. pp.~1--6. {IEEE} (2022)

\bibitem{DBLP:conf/atva/Neider12}
Neider, D.: Computing minimal separating {DFA}s and regular invariants using
  {SAT} and {SMT} solvers. In: {ATVA} 2012. Lecture Notes in Computer Science,
  vol.~7561, pp. 354--369. Springer (2012)

\bibitem{Neider14}
Neider, D.: Applications of automata learning in verification and synthesis.
  Ph.D. thesis, {RWTH} Aachen University (2014)

\bibitem{UlyantsevZS15}
Ulyantsev, V., Zakirzyanov, I., Shalyto, A.: {BFS}-based symmetry breaking
  predicates for {DFA} identification. In: {LATA} 2015. Lecture Notes in
  Computer Science, vol.~8977, pp. 611--622. Springer (2015)

\bibitem{UlyantsevZS16}
Ulyantsev, V., Zakirzyanov, I., Shalyto, A.: Symmetry breaking predicates for
  {SAT}-based {DFA} identification. CoRR  \textbf{abs/1602.05028} (2016)

\bibitem{Vaandrager17}
Vaandrager, F.W.: Model learning. Commun. {ACM}  \textbf{60}(2),  86--95 (2017)

\bibitem{ZakirzyanovMIUM19}
Zakirzyanov, I., Morgado, A., Ignatiev, A., Ulyantsev, V., Marques{-}Silva, J.:
  Efficient symmetry breaking for {SAT}-based minimum {DFA} inference. In:
  {LATA} 2019. Lecture Notes in Computer Science, vol. 11417, pp. 159--173.
  Springer (2019)

\bibitem{ZakirzyanovSU17}
Zakirzyanov, I., Shalyto, A., Ulyantsev, V.: Finding all minimum-size {DFA}
  consistent with given examples: {SAT}-based approach. In: {SEFM} 2017
  Workshops. Lecture Notes in Computer Science, vol. 10729, pp. 117--131.
  Springer (2017)

\end{thebibliography}

%
\newpage
\appendix

\section{Proofs of Theorems and Lemmas} \label{appendix:proofs}

\restateDFACorrectness*
\begin{proof}
    The proof is by induction on the 3DFA construction steps that the construction preserves the following invariant: 
    the language of each node mapped by the $\mathtt{Register}$ to the same representative state is the same, that is, for all nodes $v, v'$, if $\mathtt{Register}(v) = \mathtt{Register}(v')$, then $\lang(v) = \lang(v')$.

    For the base case, at the end of the initialization phase, $\mathtt{Register}$ contains two types of representative states: 
    one single ``accepting'' representative state for all accepting nodes without outgoing transitions and one different ``rejecting'' representative state for each rejecting node.
    Obviously, the invariant is satisfied after the initialization phase, since all accepting nodes mapped to the only ``accepting'' representative state have $\setnocond{\epsilon}$ as language, given that they have no outgoing transitions; 
    each rejecting node is mapped to a different ``rejecting'' representative state, so $\mathtt{Register}(v) = \mathtt{Register}(v')$ holds only when $v' = v$ for each rejecting node~$v$.

    For the induction case, assume that the invariant holds for the current \texttt{Register}; 
    we need to prove that it holds also after $\mathtt{Register}$ has been updated to $\mathtt{Register}'$ by one iteration of the reduction process.
    By the reduction step, $\mathtt{Register}'$ is $\mathtt{Register}$ extended with new mappings, namely, for each equivalence class $C$ identified by the reduction step, all nodes in $C$ are mapped to the same newly created representative state.
    To show that the invariant is preserved it suffices to consider pairs of equivalent nodes (i.e., nodes belonging to the same equivalence class), since nodes in different classes are mapped to different newly created representative states.
    Thus, consider an equivalence class and two nodes $v$ and $v'$ in it.
    By definition of the reduction step, we have that
    \begin{itemize}
    \item 
        $\epsilon \in \lang(v) \iff \epsilon \in \lang(v')$, since either both $v$ and $v'$ are accepting (thus $\epsilon \in \lang(v)$ and $\epsilon \in \lang(v')$), or both are don't-care (thus $\epsilon \notin \lang(v)$ and $\epsilon \notin \lang(v')$);
    \item
        for every letter $a \in \Sigma$, we have two cases: 
        either both $v$ and $v'$ have no $a$-successor, or both have an $a$-successor (denoted $v_{a}$ and $v'_{a}$, respectively). 
        In the former case, for each word $u \in \Sigma^{*}$, we have $a \cdot u \notin \lang(v)$ and $a \cdot u \notin \lang(v')$. 
        In the latter case, the reduction step enforces $\mathtt{Register}(v_{a}) = \mathtt{Register}(v'_{a})$ and the inductive hypothesis ensures $\lang(v_{a}) = \lang(v'_{a})$. 
        This implies that for each $u \in \lang(v_{a}) = \lang(v'_{a})$, we have $a \cdot u \in \lang(v)$ and $a \cdot u \in \lang(v')$ and for each $u \in \Sigma^{*} \setminus \lang(v_{a}) = \Sigma^{*} \setminus \lang(v'_{a})$, we have $a \cdot u \notin \lang(v)$ and $a \cdot u \notin \lang(v')$.
    \end{itemize}
    Since the above cases cover all possible words in $\Sigma^{*}$, it follows that $\lang(v) = \lang(v')$ as required.
    \qed
\end{proof}

\restateLemma*
\begin{proof}
Let $u, u' \in \pref(S)$ be two different prefixes and assume that $u \in \pref(S^{-})$;
the case $u' \in \pref(S^{-})$ is analogous.
We prove the statement of the lemma by contradiction:
assume that $\trans(\init, u) = \trans(\init, u')$.
Since $u \in \pref(S^{-})$ by assumption, this implies that there exists a rejecting word $uy \in S^{-}$ with $y \in \Sigma^{*}$, that is, $r = \trans(\init, uy)$ is a rejecting state.
It then follows that $r$ is also reached by $u'v$, $\trans(\init, u'y) = \trans(\trans(\init, u'), y) = \trans(\trans(\init, u), y) = r$, since $\A$ is deterministic and $\trans(\init, u) = \trans(\init, u')$.
This then entails that $u'y$ is also a rejecting word in $S^{-}$; 
however, since $uy, u'y \in S^{-}$ and both reach $r$, we get a contradiction with Theorem~\ref{thm:dfacorrectness}, which guarantees that two different rejecting words are associated with different states by our construction.
\qed
\end{proof}

\restateEncodingCorrectness*
\begin{proof}
    We assume that the SAT encoding via APTA in~\cite{LaufferYVSS22} is correct. 
    Compared to it, there are two main differences in our improved encoding: (1) replacing APTA with 3DFA, and (2) replacing Constraint~\hyperref[constraint:O1]{\textbf{O1}} with Constraint~\hyperref[constraint:O1prime]{\textbf{O1'}}.
    For the first improvement, Theorem~\ref{thm:dfacorrectness} guarantees that the constructed 3DFA is consistent with the examples $S$ as the original APTA does.
    For the second one, Constraint \hyperref[constraint:O1prime]{\textbf{O1'}} allows a merged representative state to be associated with several states of each individual DFA in the DFA decomposition and ensures every rejecting state reached by a negative example in the 3DFA can be associated with at most one state $i$ of each individual DFA in the DFA decomposition. 
    The correctness of Constraint \hyperref[constraint:O1prime]{\textbf{O1'}} is guaranteed by Lemma~\ref{lem:prefix-unique}. 
    Therefore, the correctness of Theorem~\ref{thm:encoding_correctness} is proved. 
    \qed
\end{proof}

\restateTerminationCorrectness*
\begin{proof}
We first consider the termination of Algorithm~\ref{alg:states_optimal_DIP}. It is not hard to find that a correct decomposition $\mathcal{D}={\A_1,\dots,\A_n}$ can always be found when $N=2+\sum_{u\in S^-}{(|u|+2)}$ where a 2-states DFA accepting all words and every other DFA rejecting a different negative example $u\in S^-$. Therefore, the loop iterations in Algorithm~\ref{alg:states_optimal_DIP} are bounded by $1 + \sum_{u \in S^{-}} {(|u|+2)}$, as $N$ starts with $2$ and it is increased by $1$ every round. Then we have the termination of Algorithm~\ref{alg:states_optimal_DIP}. As Algorithm~\ref{alg:computeDFAcombinations} enumerates all possible states allocations in the states-optimal order, the correctness of Algorithm~\ref{alg:states_optimal_DIP} is straightforward from Theorem~\ref{thm:encoding_correctness} and the correctness of the SAT solver. 
    \qed
\end{proof}

    


\section{SAT Encoding of DFA-DIP and Algorithm for Solving Pareto-optimal DIP in~\cite{LaufferYVSS22}}
\label{sec:sat_encoding}
Below, we list the complete SAT encoding via APTA utilized in~\cite{LaufferYVSS22}. The encoding extends the SAT encoding for monolithic DFA identification. We refer to the constraints as follows:
\begin{enumerate}
    \item A positive example must be accepted by all DFAs:
    \[
    \bigwedge_{v \in V^{+}} \bigwedge_{k \in [n]} \bigwedge_{i \in [m_{k}]} x_{v,i}^{k} \implies z_{i}^{k}.
    \]
    \item A negative example must be rejected by at least one DFA:
    \[
    \bigwedge_{v \in V^{-}} \bigwedge_{k \in [n]} \bigwedge_{i \in [m_{k}]} x_{v,i}^{k} \implies \neg z_{i}^{k}.
    \]
    \item Each state of APTA has at least one color for each DFA:\footnote{In the CoRR version of~\cite{LaufferYVSS22}, the formula is written with the conjunction $\bigwedge_{i \in [m_{k}]}$ instead of the disjunction $\bigvee_{i \in [m_{k}]}$. However, to make the constraint satisfiable, this would cause all variables $x_{v,i}^{k}$ to be set to true, which is not what the constraint should require. The tool implementation indeed uses $\bigvee_{i \in [m_{k}]}$, so here we provide the corrected version of the constraint.}
    \[
    \bigwedge_{v \in V} \bigwedge_{k \in [n]} \bigvee_{i \in [m_{k}]} x_{v,i}^{k}.
    \]
    \item A transition of a DFA is set when a state and its parent are both colored:
    \[
    \bigwedge_{v \in V \setminus \{v_r\}} \bigwedge_{k \in [n]} \bigwedge_{i,j \in [m_k]} (x_{p(v),i}^{k} \land x_{v,j}^{k}) \implies y_{l(v),i,j}^{k}.
    \]
    \item A transition of a DFA targets at most one state:
    \[
    \bigwedge_{l \in \Sigma} \bigwedge_{k \in [n]} \bigwedge_{\substack{i,j,t \in [m_{k}] \\ j < t}} y_{l,i,j}^{k} \implies \neg y_{i,j,t}^{k}.
    \]
    \item Each state of APTA has at most one color for each DFA:
    \[
    \bigwedge_{v \in V} \bigwedge_{k \in [n]} \bigwedge_{i,j \in [m_k]} \neg x_{v,i}^{k} \lor \neg x_{v,j}^{k,j}.
    \]
    \item A transition of a DFA targets at least one state:
    \[
    \bigwedge_{l \in \Sigma} \bigwedge_{k \in [n]} \bigwedge_{i,j \in [m_k]} y_{l,i,j}^{k} .\]
    \item For each DFA, a node color is set when the color of the parent node and the transition between them are set:
    \[
    \bigwedge_{v \in V \setminus \{v_r\}} \bigwedge_{k \in [n]} \bigwedge_{i,j \in [m_k]} (x_{p(v),i}^{k} \land y_{l(v),i,j}^{k}) \implies x_{v,j}^{k}.
    \]
    \item Accepting-rejecting nodes of APTA cannot be merged:
    \[
    \bigwedge_{v^{-} \in V^{-}} \bigwedge_{v^{+} \in V^{+}} \bigwedge_{k \in [n]} \bigwedge_{i \in [m_{k}]} \left( x^{k}_{v^{-},i} \wedge \neg z^{k}_{i} \right) \implies \neg x^{k}_{v^{+},i}.
    \]
\end{enumerate}
The next set of constraints encode the symmetry breaking clauses introduced in \cite{UlyantsevZS15,UlyantsevZS16} to avoid consideration of isomorphic DFAs. The main idea is to enforce  individual DFA states to be enumerated in a depth-first search (DFS) order. Let $\Sigma = \{ l_1, \ldots, l_L \}$.

\begin{enumerate}
    \item Each state must have a smaller parent in the DFS order:
    \[
    \bigwedge_{k \in [n]} \bigwedge_{i \in [2, m_{k}]} \left( p^{k}_{i,1} \vee \cdots \vee p^{k}_{i,i-1} \right).
    \]

    \item Define $p^{k}_{j,i}$ in terms of auxiliary variable $t^{k}_{i,j}$:
    \[
    \bigwedge_{k \in [n]} \bigwedge_{\substack{i,j \in [m_{k}] \\ i < j}} \left( p^{k}_{j,i} \iff t^{k}_{i,j} \wedge t^{k}_{i+1,j} \wedge \cdots \wedge t^{k}_{j-1,j} \right).
    \]

    \item Define $t^{k}_{i,j}$ in terms of $y_{l,j,j}$:
    \[
    \bigwedge_{k \in [n]} \bigwedge_{\substack{i,j \in [m_{k}] \\ i < j}} \left( t^{k}_{i,j} \iff y^{k}_{l_1,i,j} \vee \cdots \vee y^{k}_{l_L,i,j} \right).
    \]

    \item The parent relationship follows the DFS order:
    \[
    \bigwedge_{k \in [n]} \bigwedge_{\substack{i,j,p,q \in [m_{k}] \\ i < p < j < q}} \left( p^{k}_{j,i} \implies \neg t^{k}_{p,q} \right).
    \]

    \item Define $m^{k}_{l,i,j}$ in terms of $y^{k}_{l,i,j}$:
    \[
    \bigwedge_{k \in [n]} \bigwedge_{\substack{i,j \in [m_{k}] \\ i < j}} \bigwedge_{\substack{l_r \in \Sigma}} \left( m^{k}_{l_r,i,j} \iff y^{k}_{l_r,i,j} \wedge \cdots \wedge y^{k}_{l_1,i,j} \right).
    \]

    \item Enforce DFAs to be DFS-enumerated in the order of symbols on transitions:
    \[
    \bigwedge_{k \in [n]} \bigwedge_{\substack{i,j,q \in [m_{k}] \\ i < j < q}} \bigwedge_{\substack{l_r,l_s \in \Sigma \\ r < s}} \left( p^{k}_{j,i} \wedge p^{k}_{q,i} \wedge m^{k}_{l_r,i,j} \implies \neg m^{k}_{l_s,i,k} \right).
    \]
\end{enumerate}

Based on the above encoding, Lauffer \emph{et al.} provided Algorithm~\ref{alg:Pareto_DIP_APTA} in~\cite{LaufferYVSS22}. 
Compared to it, our method for solving Pareto-optimal DIP replaces the SAT encoding via APTA  by our improved encoding via 3DFA in Line~\ref{alg3_line:solving} of Algorithm~\ref{alg:Pareto_DIP_APTA}.

\begin{algorithm}[H]
    \caption{Pareto-optimal DIP Solving~\cite{LaufferYVSS22}}
    \label{alg:Pareto_DIP_APTA}
    \SetKwInput{Input}{Input}\SetKwInOut{Output}{Output}\SetNoFillComment
    \Input{The labeled examples $S=\{S^+,S^-\}$; the number of DFAs $n$.}
    \Output{Pareto frontier $P^{*}$.}
        $(P^{*},Q) \gets \{(2, \ldots, 2)\}$\tcp*{Initial Pareto frontier and queue} 
        \While{$Q \neq \emptyset$}{
            $m \gets Q.\text{dequeue}()$\; 
            \If{$\nexists \hat{m} \in P^{*}$ such that $\hat{m} \prec m$}{
                $\mathit{SAT}, \A \gets \textsc{Solve}(n, m, S_+, S_-)$\; \label{alg3_line:solving}
                \If{$\mathit{SAT}$}{
                    $P^{*} \gets P^{*} \cup \A$\tcp*{Update Pareto frontier}
                }
                \Else{
                    \For{$k \gets 1$ \KwTo $n$}{
                        $(m',m'_k) \gets (m,m'_k + 1)$\;
                        \If{$\textup{ordered}(m')$}{
                            $Q.\text{enqueue}(m')$\;
                        }
                    }
                }
            }
        }
        \Return{$P^{*}$}\;
\end{algorithm}

\section{More Detailed Experimental Results}
\label{appendix:experiments}

In this appendix, we provide more detailed plots and analyses for the experiments we presented in Section~\ref{sec:experiments}.

\subsection{Distribution of the benchmark files}

\begin{table}
    \caption{Distribution of the benchmarks with respect to the number of examples, the maximum length of each example, and the alphabet size; entries in the table show for which alphabet size there are 10 benchmarks}
    \label{tab:expBenchmarkDistribution}
    \centering
    \setlength{\tabcolsep}{2pt}
    \begin{tabular}{c|cccccccc}
        \toprule
        \multirow{2}{*}{\#Examples} & \multicolumn{8}{c}{Max length} \\
        & 3 & 4 & 5 & 6 & 7 & 8 & 9 & 10 \\
        \hline
        5 & 2/4/6/8 & 2/4/6/8 & 2/4/6/8 & 2/4/6/8 & 2/4/6/8 & 2/4/6/8 & 2/4/6/8 & 2/4/6/8 \\
        10 & -/4/6/8 & 2/4/6/8 & 2/4/6/8 & 2/4/6/8 & 2/4/6/8 & 2/4/6/8 & 2/4/6/8 & 2/4/6/8 \\
        15 & -/4/6/8 & -/4/6/8 & 2/4/6/8 & 2/4/6/8 & 2/4/6/8 & 2/4/6/8 & 2/4/6/8 & 2/4/6/8 \\
        20 & -/4/6/8 & -/4/6/8 & -/4/6/8 & 2/4/6/8 & 2/4/6/8 & 2/4/6/8 & 2/4/6/8 & 2/4/6/8 \\
        25 & -/-/6/8 & -/4/6/8 & -/4/6/8 & 2/4/6/8 & 2/4/6/8 & 2/4/6/8 & 2/4/6/8 & 2/4/6/8 \\
        30 & -/-/6/8 & -/4/6/8 & -/4/6/8 & -/4/6/8 & 2/4/6/8 & 2/4/6/8 & 2/4/6/8 & 2/4/6/8 \\
        35 & -/-/6/8 & -/4/6/8 & -/4/6/8 & -/4/6/8 & -/4/6/8 & 2/4/6/8 & 2/4/6/8 & 2/4/6/8 \\
        40 & -/-/6/8 & -/4/6/8 & -/4/6/8 & -/4/6/8 & -/4/6/8 & 2/4/6/8 & 2/4/6/8 & 2/4/6/8 \\
        45 & -/-/6/8 & -/4/6/8 & -/4/6/8 & -/4/6/8 & -/4/6/8 & -/4/6/8 & 2/4/6/8 & 2/4/6/8 \\
        50 & -/-/6/8 & -/4/6/8 & -/4/6/8 & -/4/6/8 & -/4/6/8 & -/4/6/8 & 2/4/6/8 & 2/4/6/8 \\
        60 & -/-/-/8 & -/4/6/8 & -/4/6/8 & -/4/6/8 & -/4/6/8 & -/4/6/8 & -/4/6/8 & 2/4/6/8 \\
        70 & -/-/-/8 & -/4/6/8 & -/4/6/8 & -/4/6/8 & -/4/6/8 & -/4/6/8 & -/4/6/8 & -/4/6/8 \\
        80 & -/-/-/8 & -/4/6/8 & -/4/6/8 & -/4/6/8 & -/4/6/8 & -/4/6/8 & -/4/6/8 & -/4/6/8 \\
        90 & -/-/-/8 & -/4/6/8 & -/4/6/8 & -/4/6/8 & -/4/6/8 & -/4/6/8 & -/4/6/8 & -/4/6/8 \\
        100 & -/-/-/- & -/4/6/8 & -/4/6/8 & -/4/6/8 & -/4/6/8 & -/4/6/8 & -/4/6/8 & -/4/6/8 \\
        \bottomrule
    \end{tabular}
\end{table}

Table~\ref{tab:expBenchmarkDistribution} shows how benchmarks are distributed among the different choices of number of examples, maximum length of each example, and alphabet size.
Each entry in the table says where there are 10 benchmarks of the corresponding combination: 
for example, entry ``-/4/6/8'' at ``\#Examples'' 10 and ``Max length'' 3 means that there are 10 benchmarks are available for all alphabet sizes except for $|\Sigma| = 2$. 
This is because with 2 letters and words of length at most 3, we can have at most $1 + 2 + 4 + 8 = 15$ words, and these are not enough to have 10 positive examples and 10 negative examples in the benchmarks; 
with larger alphabets, instead, we have enough words to choose randomly 10 accepted and 10 rejected words to populate the example sets.
This is why we have a different total number of benchmarks for the different alphabet sizes. 
Note that we also consider alphabet size 10 but we don not report it in the table since all entries would have it;
alphabet of size 10 also enables to have 10 benchmarks of maximum length 2 and 5 examples in each of the positive and negative example sets.
Since no other combination of alphabet size and number of examples allows us to get benchmarks with maximum length 2, we omit the corresponding column from the table.

\subsection{Runtime comparison between \toolOriginal and \toolLaufferOurEncoding}

\begin{figure}
    \resizebox{\linewidth}{!}{\includegraphics{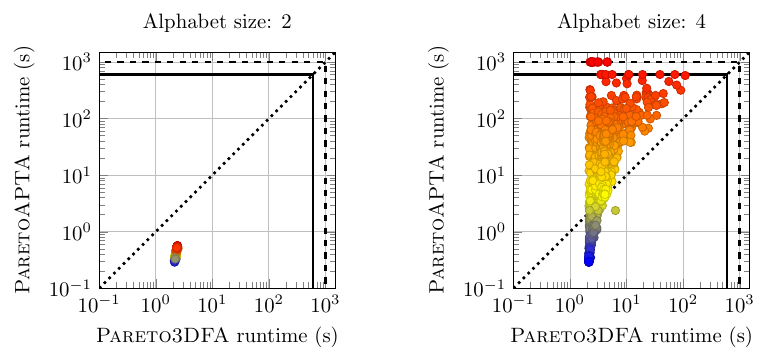}}
    \caption{Running time comparison between \toolOriginal and \toolLaufferOurEncoding using two DFAs}
    \label{fig:expScatterLauffer2DFAs}
\end{figure}

\begin{figure}[th]
    \resizebox{\linewidth}{!}{\includegraphics{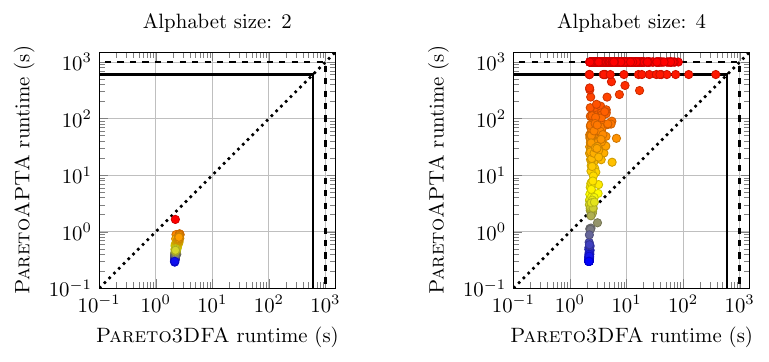}}
    \caption{Running time comparison between \toolOriginal and \toolLaufferOurEncoding using three DFAs}
    \label{fig:expScatterLauffer3DFAs}
\end{figure}

\begin{figure}
    \resizebox{\linewidth}{!}{\includegraphics{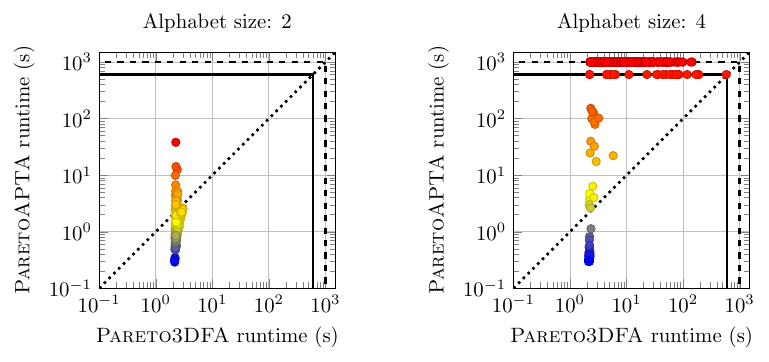}}
    \caption{Running time comparison between \toolOriginal and \toolLaufferOurEncoding using four DFAs}
    \label{fig:expScatterLauffer4DFAs}
\end{figure}

\begin{figure}
    \resizebox{\linewidth}{!}{\includegraphics{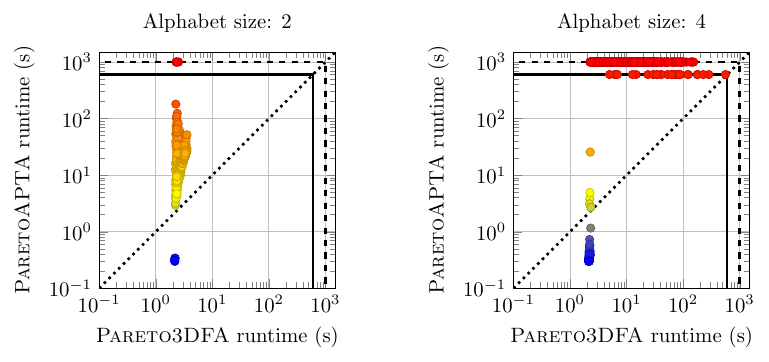}}
    \caption{Running time comparison between \toolOriginal and \toolLaufferOurEncoding using five DFAs}
    \label{fig:expScatterLauffer5DFAs}
\end{figure}

\begin{figure}
    \resizebox{\linewidth}{!}{\includegraphics{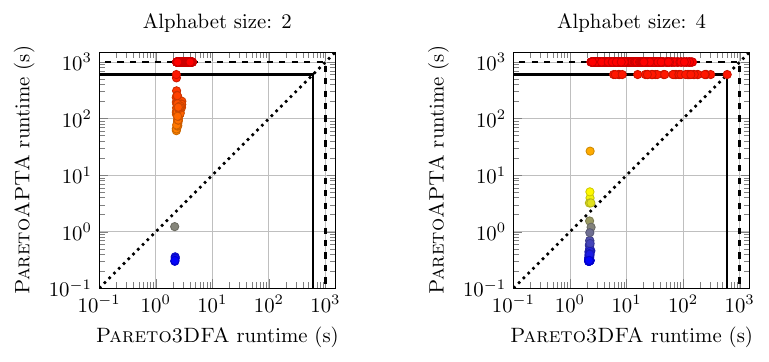}}
    \caption{Running time comparison between \toolOriginal and \toolLaufferOurEncoding using six DFAs}
    \label{fig:expScatterLauffer6DFAs}
\end{figure}

Fig.~\ref{fig:expScatterLauffer2DFAs} shows the experiments on running \toolOriginal and \toolLaufferOurEncoding on the 460 benchmarks with $|\Sigma| = 2$ and 1050 benchmarks with $|\Sigma| = 4$ when using a decomposition with two DFAs.
As we can see from the plots, for $|\Sigma| = 2$, \toolOriginal solved all cases by taking between 0.3 and 0.6 seconds;
\toolLaufferOurEncoding instead took between 2.1 and 2.5 seconds.
These larger running times can be explained by the additional computation needed in the reduction of the 3DFA obtained by merging nodes that are equivalent (cf.~Section~\ref{subsec:3dfa_encoding}); 
the time spent in this reduction is not compensated by the time saved by the SAT solver to decide the satisfiability of the smaller encoding formulas.
The situation however changes when we consider the benchmarks with $|\Sigma| = 4$:
here, the more demanding a benchmark is, the more effective the 3DFA encoding becomes, as shown by the cloud of points above the diagonal.

Fig.~\ref{fig:expScatterLauffer3DFAs} is similar to Fig.~\ref{fig:expScatterLauffer2DFAs}, but in this case we set the number of DFAs in the decomposition to be three.
Here, the running time of \toolLaufferOurEncoding for $|\Sigma| = 2$ is essentially the same (between 2.1 and 2.7 seconds);
\toolOriginal instead starts requiring more time (between 0.3 and 1.7 seconds).
The increased demand by \toolOriginal is more evident from the plot relative to $|\Sigma| = 4$, where the cloud of points is narrower and several points occur on the timeout/memoryout lines.  

The sensibility of \toolOriginal to the number of DFAs in the decomposition is more and more clear when we increase them to four, five, and six, as shown by the plots in Fig.~\ref{fig:expScatterLauffer4DFAs}, \ref{fig:expScatterLauffer5DFAs}, and~\ref{fig:expScatterLauffer6DFAs}, respectively.
In these plots, more and more points approach the timeout/memoryout lines for \toolOriginal, also when $|\Sigma| = 2$; 
\toolLaufferOurEncoding instead is less affected by the number of DFAs.

\subsection{Behavior of \toolLaufferOurEncoding when varying the number of DFAs}

\begin{figure}
    \centering
    \includegraphics{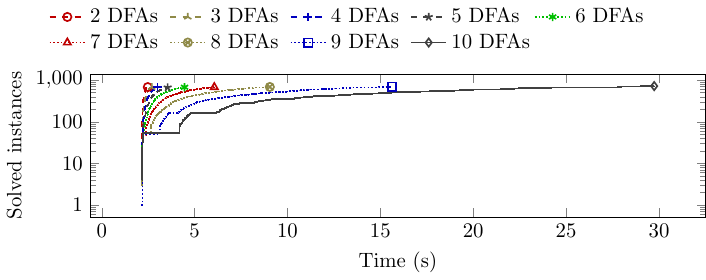}
    \caption{Cactus plot for the solved cases vs.\@ time by \toolLaufferOurEncoding on benchmarks with $|\Sigma| = 2$}
    \label{fig:expCactusLOETime2}
\end{figure}

\begin{figure}
    \centering
    \includegraphics{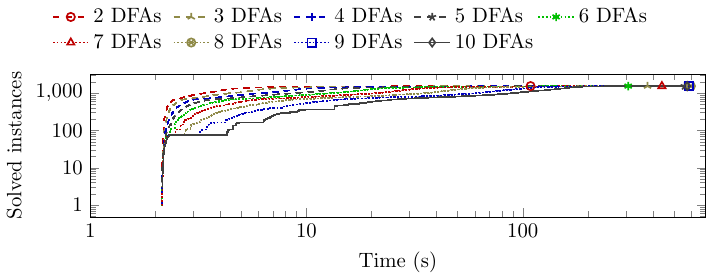}
    \caption{Cactus plot for the solved cases vs.\@ time by \toolLaufferOurEncoding on benchmarks with $|\Sigma| = 4$}
    \label{fig:expCactusLOETime4}
\end{figure}

In Fig.~\ref{fig:expCactusLOETime2} and~\ref{fig:expCactusLOETime4} we show the cactus plots for the solved cases vs.\@ time by \toolLaufferOurEncoding on benchmarks with $|\Sigma| = 2$ and $|\Sigma| = 4$, respectively, where we ran \toolLaufferOurEncoding with the number of DFAS in the decomposition ranging between two and ten. 
A point $(x, y)$ in the plot means that there have been $y$ instances that have taken at most $x$ seconds each to be analyzed successfully.
We also add a marker corresponding to the instance taking the longest time for a given number of DFAs in the decomposition.

As one would expect, increasing the number of DFAs causes \toolLaufferOurEncoding to run for longer time before solving the problem.
This is due to two main factors: 
with more DFAs, there might be more decompositions that need to be checked before finding the minimal one under the Pareto-optimal partial order;
with more DFAs, the formula obtained from encoding the 3DFA is larger and with more variables, thus the SAT solver is likely to require more time before deciding its satisfiability.
This is evident in the plot in Fig.~\ref{fig:expCactusLOETime2} about the 460 benchmarks with $|\Sigma| = 2$:
the point corresponding to the instance that took the highest amount of time to be solved moves steadily to the right of the plot when we increase the number of DFAs from two to ten.
This is slightly different for the plot in Fig.~\ref{fig:expCactusLOETime4} about the 460 benchmarks with $|\Sigma| = 4$:
for instance, the highlighted point for six DFAs is on the left of the points for three, four, and five DFAs. 
This is caused by the fact that on six DFAs, \toolLaufferOurEncoding went timeout in one case, the one taking the longest for three, four, and five DFAs.

\subsection{Behavior of \toolOurStateOptimalDIP on varying the size of the alphabet}

\begin{figure}[th]
    \centering
    \includegraphics{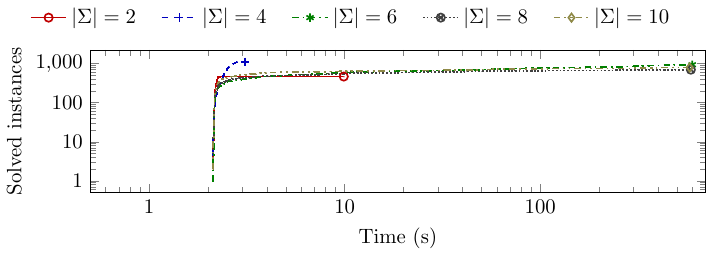}
    \caption{Cactus plot for the solved cases vs.\@ time by \toolOurStateOptimalDIP}
    \label{fig:expCactusSODIPtime}
\end{figure}

\begin{figure}[th]
    \centering
    \includegraphics{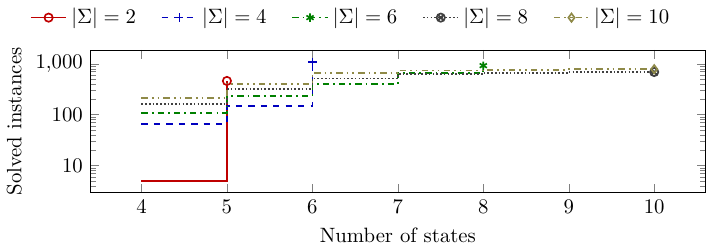}
    \caption{Cactus plot for the solved cases vs.\@ states by \toolOurStateOptimalDIP}
    \label{fig:expCactusSODIPstates}
\end{figure}

We now focus our attention on the behavior of \toolOurStateOptimalDIP when we change the size of the alphabet.
In Fig.~\ref{fig:expCactusSODIPtime} we show the cactus plots for the solved cases vs.\@ the running time while in Fig.~\ref{fig:expCactusSODIPstates} we consider the states;
these two plots present in a different format the same outcomes shown by the box plots in Fig.~\ref{fig:expBoxesSODIP}.

As already reported in Section~\ref{subsec:experimentsSODIP}, \toolOurStateOptimalDIP solved successfully all 460 benchmarks for $|\Sigma| = 2$ and 1090 benchmarks for $|\Sigma| = 4$; 
there have been experiments that did not complete for the other alphabet sizes:
out of 1150 benchmarks with $|\Sigma| = 6$, there have been 231 timeouts while out of 1190 benchmarks with $|\Sigma| = 8$, there have been 501 timeouts;
lastly, for the 1210 benchmarks with $|\Sigma| = 10$, there have been 427 timeouts.
No failure by memoryout occurred in any of the benchmarks.

As we can see from the plots, the larger the alphabet is, the longer \toolOurStateOptimalDIP takes to give a successful decomposition.
This is related to the fact that with more letters in the alphabet, it is more difficult to find equivalent states in the 3DFA-based encoding, so the acceptor has more states in it.
Moreover, with a larger alphabet, we need more variables to encode the transitions, so the encoding formula is larger and the SAT solver is likely to require more time to decide its satisfiability.

All these things also reflect on the number of states, since we likely need more states in the DFAs in the decomposition to capture the different languages of the states in the 3DFA-based acceptor.

\begin{table}
    \caption{Distribution of the minimal states computed by \toolOurStateOptimalDIP{} vs.\@ alphabet size}
    \label{tab:expStateDistributionSODIPalphabet}
    \centering
    \setlength{\tabcolsep}{4pt}
    \begin{tabular}{c|r|rrrrrrr|r}
        \toprule
        \multirow{2}{*}{$|\Sigma|$} & \multirow{2}{*}{\#Benchmarks} & \multicolumn{7}{c|}{Number of states $N$} & \multirow{2}{*}{Failures}\\
        & & \multicolumn{1}{c}{4} & \multicolumn{1}{c}{5} & \multicolumn{1}{c}{6} & \multicolumn{1}{c}{7} & \multicolumn{1}{c}{8} & \multicolumn{1}{c}{9} & \multicolumn{1}{c|}{10} & \\
        \midrule
        2 & 460 & 5 & 455 & 0 & 0 & 0 & 0 & 0 & 0 \\
        4 & 1090 & 66 & 81 & 943 & 0 & 0 & 0 & 0 & 0 \\
        6 & 1150 & 107 & 129 & 162 & 255 & 266 & 0 & 0 & 231 \\
        8 & 1190 & 160 & 159 & 195 & 113 & 35 & 25 & 2 & 501 \\
        10 & 1210 & 210 & 196 & 250 & 85 & 18 & 21 & 3 & 427 \\
        \midrule
        Total & 5100 & 548 & 1020 & 1550 & 453 & 319 & 45 & 5 & 1159 \\
        \bottomrule
    \end{tabular}
\end{table}
In Table~\ref{tab:expStateDistributionSODIPalphabet} we report the number of benchmarks solved by \toolOurStateOptimalDIP with the computed number of states, split by alphabet size.
From the table we can see that the larger the alphabet is, the more states are needed. 
This is also caused by the fact that with more letters available, it is possible to have more examples in the positive and negative example sets forming the benchmark $S = (S^{+}, S^{-})$ (cf.~Table~\ref{tab:expBenchmarkDistribution}).
Table~\ref{tab:expStateDistributionSODIPmaxLength} is similar to Table~\ref{tab:expStateDistributionSODIPalphabet} except for the fact that we report the number of states against the maximum length of the examples; 
in Table~\ref{tab:expStateDistributionSODIPexamples} we instead consider the number of positive and negative examples.
\begin{table}
    \caption{Distribution of the minimal states computed by \toolOurStateOptimalDIP{} vs.\@ maximum length of the examples}
    \label{tab:expStateDistributionSODIPmaxLength}
    \centering
    \setlength{\tabcolsep}{4pt}
    \begin{tabular}{c|r|rrrrrrr|r}
        \toprule
        \multirow{2}{*}{Max length} & \multirow{2}{*}{\#Benchmarks} & \multicolumn{7}{c|}{Number of states $N$} & \multirow{2}{*}{Failures} \\
        & & \multicolumn{1}{c}{4} & \multicolumn{1}{c}{5} & \multicolumn{1}{c}{6} & \multicolumn{1}{c}{7} & \multicolumn{1}{c}{8} & \multicolumn{1}{c}{9} & \multicolumn{1}{c|}{10} & \\
        \midrule
        2 & 10 & 10 & 0 & 0 & 0 & 0 & 0 & 0 & 0 \\
        3 & 440 & 63 & 78 & 101 & 84 & 60 & 46 & 5 & 3 \\
        4 & 620 & 70 & 83 & 215 & 84 & 72 & 0 & 0 & 96 \\
        5 & 630 & 65 & 108 & 199 & 46 & 45 & 0 & 0 & 167 \\
        6 & 650 & 66 & 125 & 203 & 30 & 1 & 0 & 0 & 225 \\
        7 & 660 & 69 & 133 & 203 & 43 & 1 & 0 & 0 & 211 \\
        8 & 680 & 68 & 151 & 205 & 31 & 42 & 0 & 0 & 183 \\
        9 & 700 & 72 & 162 & 214 & 61 & 58 & 0 & 0 & 133 \\
        10 & 710 & 65 & 180 & 210 & 75 & 40 & 0 & 0 & 141 \\
        \midrule
        Total & 5100 & 548 & 1020 & 1550 & 453 & 319 & 45 & 5 & 1159 \\
        \bottomrule
    \end{tabular}
\end{table}
\begin{table}
    \caption{Distribution of the minimal states computed by \toolOurStateOptimalDIP{} vs.\@ number of examples}
    \label{tab:expStateDistributionSODIPexamples}
    \centering
    \setlength{\tabcolsep}{4pt}
    \begin{tabular}{c|r|rrrrrrr|r}
        \toprule
        \multirow{2}{*}{$|S^{+}| =|S^{-}| $} & \multirow{2}{*}{\#Benchmarks} & \multicolumn{7}{c|}{Number of states $N$} & \multirow{2}{*}{Failures}\\
        & & \multicolumn{1}{c}{4} & \multicolumn{1}{c}{5} & \multicolumn{1}{c}{6} & \multicolumn{1}{c}{7} & \multicolumn{1}{c}{8} & \multicolumn{1}{c}{9} & \multicolumn{1}{c|}{10} & \\
        \midrule
        5 & 410 & 318 & 92 & 0 & 0 & 0 & 0 & 0 & 0 \\
        10 & 390 & 181 & 188 & 21 & 0 & 0 & 0 & 0 & 0 \\
        15 & 380 & 47 & 248 & 85 & 0 & 0 & 0 & 0 & 0 \\
        20 & 370 & 2 & 209 & 159 & 0 & 0 & 0 & 0 & 0 \\
        25 & 360 & 0 & 120 & 225 & 15 & 0 & 0 & 0 & 0 \\
        30 & 350 & 0 & 53 & 229 & 67 & 0 & 0 & 0 & 1 \\
        35 & 340 & 0 & 30 & 183 & 97 & 1 & 0 & 0 & 29 \\
        40 & 340 & 0 & 30 & 142 & 90 & 16 & 0 & 0 & 62 \\
        45 & 330 & 0 & 20 & 85 & 86 & 26 & 0 & 0 & 113 \\
        50 & 330 & 0 & 20 & 71 & 71 & 45 & 0 & 0 & 123 \\
        60 & 310 & 0 & 10 & 70 & 19 & 50 & 1 & 0 & 160 \\
        70 & 300 & 0 & 0 & 70 & 7 & 52 & 6 & 0 & 165 \\
        80 & 300 & 0 & 0 & 70 & 0 & 39 & 17 & 1 & 173 \\
        90 & 300 & 0 & 0 & 70 & 1 & 43 & 18 & 1 & 167 \\
        100 & 290 & 0 & 0 & 70 & 0 & 47 & 4 & 3 & 166 \\
        \midrule
        Total & 5100 & 548 & 1020 & 1550 & 453 & 319 & 45 & 5 & 1159 \\
        \bottomrule
    \end{tabular}
\end{table}

\end{document}